\documentclass[a4paper,UKenglish]{lipics}%

\usepackage{ifthen}
\newcommand{\techRep}{true} %% switch here between true and false
\newcommand{\iftechrep}{\ifthenelse{\equal{\techRep}{true}}}

\usepackage{amsmath}
\usepackage{amssymb}
\usepackage{tikz}
\usepackage{stmaryrd}
\usepackage[center]{subfigure}
\usepackage{btran}
\usepackage{tikz}
\usetikzlibrary{calc,positioning,shapes,arrows,3d,fadings,decorations.text, intersections}
\tikzstyle{state}=[draw,circle,inner sep=0.5mm, minimum size=6mm]
\tikzstyle{distr}=[draw,fill,circle,minimum size=2mm,inner sep=0mm]
\tikzstyle{arr}=[-latex']
\tikzstyle{ran}=[rounded corners,thick,draw,minimum size=1.4em,inner sep=.5ex]
\tikzstyle{tran}=[thick,draw,->,>=stealth]

\usepackage{rotating}

\newcommand{\Acc}{\mathit{Acc}}

\newcommand{\INC}{\mathsf{INC}}
\newcommand{\distINC}{\mathsf{dist}}

\newcommand{\dist}[1]{{\cal D}(#1)} %% probability distribution
\newcommand{\support}[1]{\mathit{support}(#1)} %% support of a probability distribution

\newcommand{\tran}[1]{\xrightarrow{#1}}

\newcommand{\N}{\mathbb{N}}
\newcommand{\Q}{\mathbb{Q}}

\newcommand{\F}{{\cal F}}

\newcommand{\Sigmac}{\Sigma_c}
\newcommand{\Sigmar}{\Sigma_r}
\newcommand{\Sigmai}{\Sigma_{\mathit int}}
\newcommand{\ac}{a_c}
\newcommand{\ar}{a_r}
\newcommand{\ai}{a_{\mathit int}}
\newcommand{\se}[1]{{[#1]}}

\newcommand{\lift}[1]{#1\mathclose\uparrow}

\newenvironment{qtheorem}[1]{%
{\mbox{}\newline\noindent\bf Theorem #1.}
\begin{itshape}%
}{%
\end{itshape}%
}

\title{Bisimilarity of Probabilistic Pushdown Automata
}
\author[1]{Vojt\v{e}ch Forejt}
\author[2]{Petr Jan\v{c}ar}
\author[1]{Stefan Kiefer}
\author[1]{James Worrell}
\affil[1]{Department of Computer Science, University of Oxford, UK}
\affil[2]{Department of Computer Science, FEI,
Techn.~Univ.~Ostrava, Czech Republic}

\authorrunning{V. Forejt and P. Jan\v{c}ar and S. Kiefer and J. Worrell}

\subjclass{F.3.1 Specifying and Verifying and Reasoning about Programs}% mandatory: Please choose ACM 1998 classifications from http://www.acm.org/about/class/ccs98-html . E.g., cite as "F.1.1 Models of Computation".
\keywords{bisimilarity, probabilistic systems, pushdown automata}% mandatory: Please provide 1-5 keywords

\begin{document}
%---------------------
\sloppy
\maketitle

\begin{abstract}
We study the bisimilarity problem for probabilistic pushdown automata
(pPDA) and subclasses thereof.  Our definition of pPDA allows both
probabilistic and non-deterministic branching, generalising the
classical notion of pushdown automata (without
$\varepsilon$-transitions).  Our first contribution is a general
construction that reduces checking bisimilarity of probabilistic
transition systems to checking bisimilarity of non-deterministic
transition systems.  This construction directly yields decidability of
bisimilarity for pPDA, as well as an elementary upper bound for the
bisimilarity problem on the subclass of probabilistic basic process
algebras, i.e., single-state pPDA.  We further show that, with careful
analysis, the general reduction can be used to prove an EXPTIME upper
bound for bisimilarity of probabilistic visibly pushdown automata.
Here we also provide a matching lower bound, establishing
EXPTIME-completeness.  Finally we prove that deciding bisimilarity of
probabilistic one-counter automata, another subclass of pPDA, is
PSPACE-complete.  Here we use a more specialised argument to obtain
optimal complexity bounds.
\end{abstract}

\section{Introduction}
Equivalence checking is the problem of determining whether two systems
are semantically identical.  This is an important question in
automated verification and, more generally, represents a line of
research that can be traced back to the inception of theoretical
computer science.  A great deal of work in this area has been devoted
to the complexity of \emph{bisimilarity} for various classes of
infinite-state systems related to grammars, such as one-counter
automata, basic process algebras, and pushdown automata,
see~\cite{Burkart00} for an overview.  We mention in particular the
landmark result showing the decidability of bisimilarity for pushdown
automata~\cite{Senizergues05}.

In this paper we are concerned with probabilistic pushdown automata
(pPDA), that is, pushdown automata with both non-deterministic and
probabilistic branching.  In particular, our pPDA generalise classical
pushdown automata without $\varepsilon$-transitions.  We refer to
automata with only probabilistic branching as \emph{fully
  probabilistic}.

We consider the complexity of checking bisimilarity for probabilistic
pushdown automata and various subclasses thereof.  The subclasses we
consider are probabilistic versions of models that have been
extensively studied in previous
works~\cite{Burkart00,SrbaVisiblyPDA:2009}.  In particular, we
consider probabilistic one-counter automata (pOCA), which are
probabilistic pushdown automata with singleton stack alphabet;
probabilistic Basic Process Algebras (pBPA), which are single-state
probabilistic pushdown automata; probabilistic visibly pushdown
automata (pvPDA), which are automata in which the stack action,
whether to push or pop, for each transition is determined by the input
letter.  Probabilistic one-counter automata have been studied in the
classical theory of stochastic processes as \emph{quasi-birth-death
  processes}~\cite{EWY10}.  Probabilistic BPA seems to have been
introduced in~\cite{Brazdil08}.

%The language theory of probabilistic pushdown automata has
%been studied in~\cite{AMP99}, where their equivalence with stochastic
%context-free grammars is proved.  There is also a growing body of work
%concerning the complexity of model checking probabilistic pushdown
%automata, probabilistic one-counter machines and probabilistic
%BPA~(see, e.g.,~\cite{Brazdil08,EKM05,EYstacs05Extended}).

While the complexity of bisimilarity for finite-state probabilistic
automata is well understood~\cite{Baier96,CBW12}, there are
relatively few works on equivalence of infinite-state probabilistic
systems.  Bisimilarity of probabilistic BPA was shown decidable
in~\cite{Brazdil08}, but without any complexity bound.
In~\cite{FuKatoen11} probabilistic simulation between probabilistic
pushdown automata and finite state systems was studied.

%Recently, \cite{Kiefer12} showed that the language equivalence problem
%for probabilistic visibly pushdown automata is logspace equivalent to
%the problem of \emph{polynomial identity testing}, that is,
%determining whether a polynomial presented as an arithmetic circuit is
%identically zero.  The latter problem is known to be in coRP\@.  As we
%show below, language equivalence of probabilistic pushdown automata is
%closely related to \emph{multiplicity equivalence} of context-free
%grammars~\cite{Raz93}.  Two grammars are multiplicity equivalent if
%each word has the same number of derivation trees in each grammar.

%In this paper we consider a class of probabilistic pushdown automata
%that features both probabilistic and non-deterministic branching.  We
%refer to the subclass of automata with only probabilistic branching as
%being \emph{fully probabilistic}.

\subsection{Contribution}

The starting point of the paper is a construction that can be used to
reduce the bisimilarity problem for many classes of probabilistic
systems to the bisimilarity problem for their non-probabilistic
counterparts. The reduction relies on the observation that in the
bisimilarity problem, the numbers that occur as probabilities in a
probabilistic system can be ``encoded'' as actions in the
non-probabilistic system. This comes at the price of an exponential
blow-up in the branching size, but still allows us to establish
several new results.  It is perhaps surprising that there is a
relatively simple reduction of probabilistic bisimilarity to ordinary
bisimilarity.  Hitherto it has been typical to establish decidability
in the probabilistic case using bespoke proofs, see,
e.g.,~\cite{Brazdil08,FuKatoen11}.  Instead, using our reduction, we
can leverage the rich theory that has been developed in the
non-probabilistic case.

The main results of the paper are as follows:
\begin{itemize}
\item Using the above-mentioned reduction together with the result
  of~\cite{Senizergues05}, we show that bisimilarity for probabilistic
  pushdown automata is decidable.
\item For the subclass of probabilistic
  BPA, i.e., automata with a single control state, the same reduction
  yields a 3EXPTIME upper bound for checking bisimilarity via a doubly
  exponential procedure for bisimilarity on BPA~\cite{Burkart00} (see
  also~\cite{Jancar12}).  This improves the result
  of~\cite{Brazdil08}, where only a decidability result was given
  without any complexity bound. An EXPTIME lower bound for this problem
  follows from the recent work of~\cite{Kiefer12} for non-probabilistic systems.

\item For probabilistic visibly pushdown automata, the above reduction
  immediately yields a 2EXPTIME upper bound. However we show that with
  more careful analysis we can extract an EXPTIME upper bound.  In
  this case we also show EXPTIME-hardness, thus obtaining matching
  lower and upper bounds.

\item
For fully probabilistic one-counter automata we obtain matching lower
and upper \mbox{PSPACE} bounds for the bisimilarity problem. In both cases
the bounds are obtained by adapting constructions from the
non-deterministic case~\cite{SrbaVisiblyPDA:2009,BGJ:Concur10} rather
than by using the generic reduction described above.
\end{itemize}

\section{Preliminaries}
Given a countable set $A$, a \emph{probability distribution} on $A$ is
a function $d: A \rightarrow [0,1] \cap \Q$ (the rationals)
such that $\sum_{a\in A}
d(a) =1$.  A probability distribution is {\em Dirac} if it assigns $1$
to one element and $0$ to all the others.  The \emph{support} of a
probability distribution $d$ is the set $\support{d}:=\{ a\in A :
d(a)>0\}$.  The set of all probability distributions on $A$ is
denoted by $\dist{A}$.

%The empty word is denoted by $\varepsilon$.  For a finite word
%$w=a_1a_2\ldots a_k$, we define the \emph{length} of $w$ to be $k$, and we
%use $w(i)$ to denote $a_i$.

\subsection{Probabilistic Transition Systems.}

A \emph{probabilistic labelled transition system (pLTS)} is a tuple
$\mathcal{S} = (S,\Sigma,\tran{})$, where $S$ is a finite or countable set of
\emph{states}, $\Sigma$ is a finite input \emph{alphabet}, and
$\mathop{\tran{}} \subseteq S\times \Sigma \times {\dist{S}}$ is a
\emph{transition relation}.  We write $s\tran{a} d$ to say that
$(s,a,d)\in\mathop{\tran{}}$.
We also write $s \tran{} s'$ to say that there exists $s \tran{a} d$ with $s' \in \support{d}$.
We assume that $\mathcal{S}$ is
finitely branching, i.e., each state $s$ has finitely many transitions
$s\tran{a}d$.  In general a pLTS combines probabilistic and
non-deterministic branching.  A pLTS is said to be \emph{fully
  probabilistic} if for each state $s\in S$ and action $a \in \Sigma$
we have $s \tran{a} d$ for at most one distribution~$d$.  Given a
fully probabilistic
pLTS, we write $s\tran{a,x}s'$ to say that there
is $s\tran{a}d$ such that $d(s')=x$.

Let $\mathcal{S} = (S,\Sigma,\tran{})$ be a pLTS and $R$ an
equivalence relation on $S$.  We say that two distributions $d,d' \in
\mathcal{D}(S)$ are \emph{$R$-equivalent} if for all $R$-equivalence
classes $E$, $\sum_{s \in E} d(s) = \sum_{s \in E }d'(s)$.  We
furthermore say that $R$ is a \emph{bisimulation relation} if $s
\mathrel{R} t$ implies that for each action $a\in \Sigma$ and
each transition $s\tran{a}d$ there is a transition $t\tran{a} d'$ such
that $d$ and $d'$ are $R$-equivalent.
The union of all bisimulation
relations of $\mathcal{S}$ is itself a bisimulation relation.  This
relation is called \emph{bisimilarity} and is denoted
$\sim$~\cite{SL94}.

We also have the following inductive characterisation of bisimilarity.
Define a decreasing sequence of equivalence relations $\mathop{\sim_0}\supseteq
\mathop{\sim_1}\supseteq \mathop{\sim_2}\supseteq\cdots$ by putting $s
\sim_0 t$ for all $s,t$, and $s\sim_{n+1}t$ if and only if for all $a\in
\Sigma$ and $s\tran{a}d$ there is $t\tran{a} d'$ such that $\sum_{s
  \in E} d(s) = \sum_{s \in E}d'(s)$ for all $\sim_n$-equivalence
classes $E$.  It is then straightforward that the
sequence $\sim_n$ converges to $\sim$, i.e.,
$\bigcap_{n\in\N}\mathop{\sim_n}=\mathop{\sim}$.

%For a finite word $w\in
%\Sigma^*$ and state $s\in S$, the probability that $s$  $w$,
%denoted $\Prob{\mathcal{S}}{s}{w}$, is the sum of the probabilities of
%all runs on $w$ that start in state $s$.  We say that $s,t \in S$ are
%\emph{trace equivalent} if
%$\Prob{\mathcal{S}}{s}{w}=\Prob{\mathcal{S}}{t}{w}$ for all words $w
%\in \Sigma^*$.

%inductively defined by
%$\Prob{L}{s}{w}=1$ if $w=\varepsilon$, and $\Prob{L}{s}{w} =
%\max_{s\tran{a}d} \sum_{s'} d(s') \cdot \Prob{L}{s'}{w'}$ if $w=aw'$
%for some $a\in \Sigma$ and $w\in \Sigma^*$.

%Given a pLTS $L$ with the alphabet $\Sigma$ and given its two states
%$s_1$ and $s_2$,
%the \emph{trace equivalence} problem is to
%decide whether for all $w\in \Sigma^*$ we have $\Prob{L_1}{s_1}{w} = \Prob{L_2}%{s_2}{w}$.
%
%Given a pLTS $L$ with the alphabet $\Sigma$, the \emph{bisimulation} problem is% to decide whether there
%is a bisimulation relation $\sigma$ on the states of $L$.

\subsection{Probabilistic Pushdown Automata.}

A \emph{probabilistic pushdown automaton} (pPDA) is a tuple $\Delta =
(Q,\Gamma,\Sigma,\btran{})$ where $Q$ is a finite set of
\emph{states}, $\Gamma$ is a finite \emph{stack alphabet}, $\Sigma$ is a finite
\emph{input alphabet}, and $\mathop{\btran{}} \subseteq Q\times \Gamma
\times \Sigma \times {\dist{Q\times \Gamma^{\le 2}}}$ (with
$\Gamma^{\le 2} := \{\varepsilon\} \cup \Gamma \cup (\Gamma \times
\Gamma)$)
%\stefan{I've restricted the RHSs. Petr, please add a remark
%  in the pBPA hardness proof.}  is the \emph{transition relation}
(where $\varepsilon$ denotes the empty string).

When speaking of the \emph{size} of~$\Delta$, we assume that the
probabilities in the transition relation are given as
quotients of integers written in binary.
%We write $qX \btran{a,x} q'w$ to denote $\btran{}(q,X,a)(q,w)=x$.
A tuple $(q,X)\in Q\times \Gamma$ is called a \emph{head}.  A pPDA is
\emph{fully probabilistic} if for each head $(q,X)$ and action $a \in
\Sigma$ there is at most one distribution $d$ with $(q,X,a,d) \in \mathop{\btran{}}$.
A configuration of a pPDA is an element $(q,\beta)\in Q\times
\Gamma^*$,  and we sometimes write just $q\beta$
instead of $(q,\beta)$.  We write $qX \btran{a} d$ to denote $(q,X,a,d)
\in \mathop{\btran{}}$, that is, in a control state $q$ with $X$ at the
top of the stack the pPDA makes an $a$-transition to the distribution
$d$.  In a fully probabilistic pPDA we also write $q X \btran{a,x} r
\beta$ if $q X \btran{a} d$ and $d(r \beta) = x$.

A \emph{probabilistic basic process algebra (pBPA)} $\Delta$ is a pPDA with only one control state.
In this case we sometimes omit the control state from the representation of a configuration.
A \emph{probabilistic one-counter automaton (pOCA)} is a pPDA with a
stack alphabet containing only two symbols $X$ and $Z$, where the
transition function is restricted so that $Z$ always and only occurs
at the bottom of the stack.
A \emph{probabilistic visibly pushdown automaton (pvPDA)} is a pPDA with a partition of the actions
 $\Sigma = \Sigmar \cup \Sigmai \cup \Sigmac$
 such that for all $pX \btran{a} d$ we have:
  if $a \in \Sigmar$ then $\support{d} \subseteq Q \times \{\varepsilon\}$;
  if $a \in \Sigmai$ then $\support{d} \subseteq Q \times \Gamma$;
  if $a \in \Sigmac$ then $\support{d} \subseteq Q \times (\Gamma \times \Gamma)$.

A pPDA $\Delta = (Q,\Gamma,\Sigma,\btran{})$ generates a
pLTS~\mbox{$\mathcal{S}(\Delta)=(Q\times \Gamma^*, \Sigma, \tran{})$}
as follows.  For each $\beta \in \Gamma^*$ a rule $qX \btran{a} d$ of
$\Delta$ induces a transition $qX\beta \tran{a} d'$ in
$\mathcal{S}(\Delta)$, where $d' \in \mathcal{D}(Q \times \Gamma^*)$
is defined by $d'(p \alpha \beta) = d(p \alpha)$ for all $p\in Q$ and $\alpha \in \Gamma^*$.
Note that all configurations with the empty stack define terminating
states of $\mathcal{S}(\Delta)$.
%%
%%$(q,wv) \tran{a} [ (q_1,ww_1) {\mapsto} x_1,\ldots,
%%  (q_n,ww_n){\mapsto} x_n]$ in $\mathcal{S}(\Delta)$ iff $qv\btran{a,x} [
%%  (q_1,w_1) {\mapsto} x_1,\ldots, (q_n,w_n){\mapsto} x_n]$ in
%%$\Delta$.
%%

%The probability that $\Delta$ accepts a word $w \in \Sigma^*$ from a
%configuration $q \alpha$ is the sum of the probabilities of all executions on
%$w$ in $\mathcal{S}(\Delta)$ that end in a configuration with
%empty stack.  We denote this probability by $\Prob{\Delta}{q \alpha}{w}$.

%In this paper we study two fundamental equivalence problems for pPDA.
The \emph{bisimilarity problem} asks whether two configurations
$q_1 \alpha_1$ and $q_2 \alpha_2$ of a given pPDA $\Delta$ are bisimilar when
regarded as states of the induced pLTS $\mathcal{S}(\Delta)$.
%The \emph{language equivalence problem} asks whether two configurations
%$q_1 \alpha_1$ and $q_2 \alpha_2$ of a given fully probabilistic pPDA~$\Delta$
%accept each input word with the same probability, i.e., whether
%$\Prob{\Delta}{q_1\alpha_1}{\cdot}=\Prob{\Delta}{q_2\alpha_2}{\cdot}$.

\begin{figure}
\centering
\begin{tikzpicture}[x=2.5cm,y=1.5cm,font=\scriptsize]
\foreach \x/\l in {0/pZ,1/pXZ,2/pXXZ,3/pXXXZ,4/pXXXXZ}
   \node (D\x)   at (\x,-1)   [ran] {$\l$};
%\draw [tran,loop left] (D0) to  node[left]   {$1$} (D0);
\foreach \x/\xx in {1/0,2/1,3/2,4/3}
   \draw [tran] (D\x) to  node[below]   {$0.5$} (D\xx);
\foreach \x/\l in {2/qXXZ,3/qXXXZ,4/qXXXXZ}
   \node (H\x)   at (\x,0)   [ran] {$\l$};
\foreach \x/\xx in {1/2,2/3,3/4}
%   \draw [tran,rounded corners] (D\x) -- +(.6,1) -- node[above] {$0.5$} (H\xx);
   \draw [tran,rounded corners] (D\x) -- node[very near start,above] {$0.5$} (H\xx);
\foreach \x/\xx in {2/3,3/4}
%   \draw [tran,rounded corners] (H\x) -- node[above] {$1$} +(.4,0) --  (D\xx);
   \draw [tran,rounded corners] (H\x) -- node[near start,above] {$1$} (D\xx);
\draw [thick, dotted] (H4) -- +(.8,0);
\draw [thick, dotted] (D4) -- +(.8,0);

\foreach \x/\l in {0/r,1/rX,2/rYX,3/rXXX}
   \node (A\x)   at (\x,-2)   [ran] {$\l$};
   \node (Ba)   at (2,-2.5)   [ran] {$rYX'$};
   \node (Bb)   at (3,-3)   [ran] {$rXXX'$};

   \node (Aa)   at (4,-1.9)   [ran] {$rXX$};
   \node (Ab)   at (4,-2.3)   [ran] {$rYXXX$};
   \node (Ac)   at (4,-2.7)   [ran] {$rYX'XX$};

   \draw [tran,rounded corners] (A1) -- node[midway,below] {$0.5$} (A0);
   \draw [tran,rounded corners] (A1) -- node[midway,above] {$0.3$} (A2);
   \draw [tran,rounded corners] (A1) -- node[midway,below] {$0.2$} (Ba);
   \draw [tran,rounded corners] (A2) -- node[midway,above] {$1$} (A3);
   \draw [tran,rounded corners] (A3) -- node[near end,above] {$0.5$} (Aa.west);
   \draw [tran,rounded corners] (A3) -- node[near end,above] {$0.3$} (Ab.west);
   \draw [tran,rounded corners] (A3) -- node[near end,above] {$0.2$} (Ac.west);
   \draw [tran,rounded corners] (Ba) -- node[midway,above] {$1$} (Bb);
\draw [thick, dotted] (Aa) -- +(.8,0);
\draw [thick, dotted] (Ab) -- +(.8,0);
\draw [thick, dotted] (Ac) -- +(.8,0);

\draw [thick, dotted] (Bb) -- +(.8,-0.25);

\end{tikzpicture}
\caption{A fragment of $\mathcal{S}(\Delta)$ from Example~\ref{ex-preli-example}.}
\label{fig-preli-example}
\end{figure}
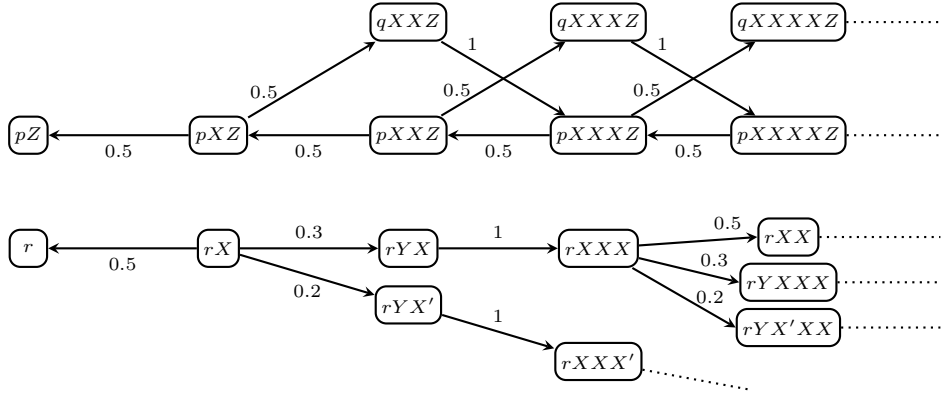

\begin{example} \label{ex-preli-example}
  Consider the fully probabilistic pPDA~$\Delta = (\{p,q,r\},\{X,X',Y,Z\},\{a\},\mathord{\btran{}})$ with the following rules
   (omitting the unique action~$a$):
  \begin{align*}
     pX &\btran{0.5} qXX, &&&
     pX &\btran{0.5} p, &
     qX &\btran{1} pXX, \\
     rX &\btran{0.3} rYX, &
     rX &\btran{0.2} rYX', &
     rX &\btran{0.5} r, &
     rY &\btran{1} rXX, \\
     rX' &\btran{0.4} rYX, &
     rX' &\btran{0.1} rYX', &
     rX' &\btran{0.5} r.
  \end{align*}
  The restriction of~$\Delta$ to the control states $p,q$ and to the stack symbols $X,Z$ yields a pOCA.
  The restriction of~$\Delta$ to the control state~$r$ and the stack symbols $X,X',Y$ yields a pBPA.
  A fragment of the pLTS $\mathcal{S}(\Delta)$ is shown in Figure~\ref{fig-preli-example}.
  The configurations $pXZ$ and~$rX$ are bisimilar,
   as there is a bisimulation relation with equivalence classes
    $\{pX^kZ\} \cup \{r w \mid w \in \{X,X'\}^k\}$ for all $k \ge 0$ and
    $\{qX^{k+1}Z\} \cup \{r Y w \mid w \in \{X,X'\}^k\}$ for all $k \ge 1$.
\end{example}

\newcommand{\stacksymb}[1]{\langle\underline{#1}\rangle}
\newcommand{\stacksymba}[1]{\langle #1 \rangle}
\section{From Probabilistic to Nondeterministic Bisimilarity}
\label{sec-prob-to-nondet}

A nondeterministic pushdown automaton (PDA) is a special case of a
probabilistic pushdown automaton in which the transition function
assigns only Dirac distributions. We give a novel reduction of the
bisimilarity problem for pPDA to the bisimilarity problem for
PDA. Because the latter is known to be decidable~\cite{Senizergues05},
we get decidability of the bisimilarity problem for pPDA.

As a first step we give the following characterisation of
$R$-equivalence of two distributions (defined earlier).
%The proof can
%be found in Appendix~\ref{app:bisim-as-game}.

\begin{lemma}\label{lem:bisim-as-game}
Let $R$ be an equivalence relation on a set $S$.  Two distributions
$d,d'$ on $S$ are $R$-equivalent if and only if for all $A \subseteq S$ we have $d(A)
\leq d'(R(A))$, where $R(A)$ denotes the image of $A$ under $R$.
\end{lemma}
\begin{proof}
For the \emph{if} direction we reason as follows.  For each
equivalence class $E$ we have $d(E) \leq d'(E)$.  But since $d$ and
$d'$ have total mass $1$ we must have $d(E)=d'(E)$ for all equivalence
classes $E$.

Conversely if $d$ and $d'$ are $R$-equivalent.  Then $d(A) \leq
d(R(A)) = d'(R(A))$ for any set $A$, since $R(A)$ is a countable union
of equivalence classes.
\end{proof}

We now give our reduction.  Let $\Delta = (Q,\Gamma,\Sigma,\btran{})$
be a pPDA and $q_1\gamma_1$, $q_2\gamma_2$ two configurations of
$\Delta$.  We define a new PDA $\Delta' =
(Q,\Gamma',\Sigma',\ctran{})$ that extends $\Delta$ with extra stack
symbols, input letters and transition rules.  In particular, a
configuration of $\Delta$ can also be regarded as a configuration of
$\Delta'$.  The definition of $\Delta'$ is such that two
$\Delta$-configurations $q_1 \gamma_1$ and $q_2\gamma_2$ are bisimilar
in $\Delta$ if and only if the same two configurations are bisimilar in
$\Delta'$.

Intuitively we eliminate probabilistic transitions by treating
probabilities as part of the input alphabet.  To this end, let $W
\subseteq \mathbb{Q}$ be the set of rational numbers of the form
$d(A)$ for some rule $ pX \btran{a} d$ in $\Delta$ and $A \subseteq
\mathit{support}(d)$.  Think of $W$ as the set of relevant transition
weights.

We define $\Delta'$ as follows.  Note that when defining rules of
$\Delta'$ we write just $q\gamma$ instead of the Dirac distribution
assigning $1$ to $q\gamma$.

\begin{itemize}
 \item The stack alphabet $\Gamma'$ contains all symbols from
   $\Gamma$. In addition, for every rule $pX \btran{a} d$ in $\Delta$
   it contains a new symbol $\stacksymba{d}$ and for every $T
   \subseteq \mathit{support}(d)$ a symbol $\stacksymba{T}$.

 \item
The input alphabet $\Sigma'$ is equal to $\Sigma\cup W\cup \{\#\}$
where $\#$ is a distinguished action not in $\Sigma$ or $W$.

 \item
The transition function $\ctran{}$ is defined as follows. For every
rule $qX \btran{a} d$, there is a rule $qX\ctran{a}q\stacksymba{d}$.
We also have a rule $q\stacksymba{d} \ctran{w} q\stacksymba{T}$ if $T
\subseteq \mathit{support}(d)$ and $d(T)\geq w \in W$.  Finally, we have a
rule $q\stacksymba{T} \ctran{\#} p\alpha$ if $p\alpha \in T$.
\end{itemize}

The PDA $\Delta'$ can be constructed in time exponential in the size
of $\Delta$, and in polynomial time if the branching degree of
$\Delta$ is bounded
(i.e. if we fix a number $N$ and consider only pPDAs with branching degree at most $N$). 
\iftechrep{See Appendix~\ref{app:analysis-complexity} for the analysis.}{The analysis can be found in~\cite{techreport}.}
The correctness of the
construction is captured by the following lemma and proved
\iftechrep{in Appendix~\ref{app:nondet-constr-correct}.}{in~\cite{techreport}.}

\begin{lemma}\label{lem:nondet-constr-correct}
For any configurations $q_1\gamma_1, q_2\gamma_2$ of $\Delta$ we have
$q_1\gamma_1 \sim q_2\gamma_2$ in $\Delta$ if and only if $q_1\gamma_1
\sim q_2\gamma_2$ in $\Delta'$.
\end{lemma}

Let us show intuitively why bisimilar configurations in $\Delta$
remain bisimilar considered as configurations of $\Delta'$.  Every
computation step of $\Delta$ is simulated in three steps by
$\Delta'$. Let $q_1X_1\gamma_1$ and $q_2X_2\gamma_2$ be bisimilar
configurations of $\Delta$.  Then in $\Delta'$ a transition of
$q_1X_1\gamma_1$ to $q_1\stacksymba{d_1}\gamma_1$ can be matched by a
transition (under the same action) of $q_2X_2\gamma_2$ to $q_2\stacksymba{d_2}\gamma_2$ such
that the distributions $d_1$ and $d_2$ are $\sim$-equivalent (and
\emph{vice versa}).  In particular, by Lemma~\ref{lem:bisim-as-game},
for any set of configurations $T$ the set $T'$ obtained by saturating
$T$ under bisimilarity is such that $d_1(T) \leq d_2(T')$.  Let
$\bar{T}$ and $\bar{T}'$ respectively contain the elements of $T$ and
$T'$ from which the suffixes $\gamma_1$ and $\gamma_2$ are removed.
Then, as a second step of simulation of $\Delta$ by $\Delta'$, a transition of $q_1\stacksymba{d_1}\gamma_1$ to a state
$q_1\stacksymba{\bar{T}}\gamma_1$ with label $w\in W$ can be matched
by a transition of $\Delta'$ to $q_2\stacksymba{\bar{T}'}\gamma_2$ with the same
label (similarly any transition of $q_2\stacksymba{d_2}\gamma_2$ can
be matched by a transition of $q_1\stacksymba{d_1}\gamma_1$).
Finally, as $T$ and $T'$ contain elements from the same bisimilarity
equivalence classes, in the third step a $\#$-transition from
$q_1\stacksymba{\bar{T}}\gamma_1$ to some $q'_1\alpha_1\gamma_1$ can
be matched by a $\#$-transition of $q_2\stacksymba{\bar{T}'}\gamma_2$
to $q'_2\alpha_2\gamma_2$ such that $q'_1\alpha_1\gamma_1$ and
$q'_2\alpha_2\gamma_2$ are again bisimilar in $\Delta$ (and \emph{vice
  versa}).

The three steps are illustrated in Figure~\ref{fig:reduction}, where the successors
of the configurations $pXZ$ and $rX$ in the system $\mathcal{S}(\Delta')$ for the PDA $\Delta'$
constructed from the pPDA $\Delta$ from Example~\ref{ex-preli-example} are drawn.
\begin{figure}
\begin{center}
\begin{tikzpicture}[x=2.5cm,y=1.5cm,font=\scriptsize]
\tikzstyle{lbl}=[inner sep=1pt]
 \node (pXZ) at (0,0) [ran] {$pXZ$};
 \node (pdZ) at (0.9,0) [ran] {$p\stacksymba{\left[\hspace{-4pt}\begin{array}{c}qXX{\mapsto}0.5\\p{\mapsto}0.5\end{array}\hspace{-4pt}\right]}Z$};
 \node (pqXXpZ) at (2.9,0) [ran] {$p\stacksymba{\{qXX,p\}}Z$};
 \node (pqXXZ) at (2.9,0.6) [ran] {$p\stacksymba{\{qXX\}]}Z$};
 \node (ppZ) at (2.9,-0.6) [ran] {$p\stacksymba{\{p\}]}Z$};
 \node (pZ) at (5.1,-0.6) [ran] {$pZ$};
 \node (qXXZ) at (5.1,0.6) [ran] {$qXXZ$};

 \draw [tran,rounded corners] (pXZ) -- node[near start,above,lbl] {$a$} (pdZ);
 \draw [tran,rounded corners] (pdZ) to node[midway,above,lbl] {$W(1)$} (pqXXpZ);
 \draw [tran,rounded corners] (pdZ.15) -- node[midway, above,lbl] {$W(0.5)$} (pqXXZ);
 \draw [tran,rounded corners] (pdZ.-15) -- node[midway, below,lbl] {$W(0.5)$} (ppZ);
 \draw [tran,rounded corners] (ppZ) -- node[midway,above,lbl] {$\#$} (pZ);
 \draw [tran,rounded corners] (pqXXpZ.-5) -| node[near start,below,lbl] {$\#$} (pZ);
 \draw [tran,rounded corners] (pqXXZ) -- node[midway,below,lbl] {$\#$} (qXXZ);
 \draw [tran,rounded corners] (pqXXpZ.5) -| node[near start,above,lbl] {$\#$} (qXXZ);
 \draw [thick, dotted] (qXXZ) -- +(.3,0);

 \newcommand{\bpo}{-2.4}

 \node (rX) at (0,\bpo) [ran] {$rX$};

 \node (rd) at (0.9,\bpo) [ran] {$r\stacksymba{\left[\hspace{-4pt}\begin{array}{c}rYX{\mapsto}0.3\\rYX'{\mapsto}0.2\\r{\mapsto}0.5\end{array}\hspace{-4pt}\right]}$};

 \node (rrYXrYX'r) at (3.6,\bpo) [ran] {$r\stacksymba{\{rYX,rYX',r\}}$};
 \node (rrYXrYX') at (2.4,\bpo+0.4) [ran] {$r\stacksymba{\{rYX,rYX'\}}$};
 \node (rrYXr) at (2.4,\bpo-0.4) [ran] {$r\stacksymba{\{rYX,r\}}$};
 \node (rrYX'r) at (3.6,\bpo+0.8) [ran] {$r\stacksymba{\{rYX',r\}}$};
 \node (rrYX) at (3.6,\bpo-0.8) [ran] {$r\stacksymba{\{rYX\}}$};
 \node (rrYX') at (2.4,\bpo+1.2) [ran] {$r\stacksymba{\{rYX'\}}$};
 \node (rr) at (2.4,\bpo-1.2) [ran] {$r\stacksymba{\{r\}}$};

 \node (r) at (5.1,\bpo-1.2) [ran] {$r$};
 \node (rYX) at (5.1,\bpo) [ran] {$rYX$};
 \node (rYX') at (5.1,\bpo+1.2) [ran] {$rYX'$};

 \draw [tran,rounded corners] (rX) -- node[near start,above,lbl] {$a$} (rd);

 \draw [tran,rounded corners] (rd.-30) |- node[pos=0.6,above,lbl] {$W(0.3)$} (rrYX);
 \draw [tran,rounded corners] (rd.40) |- node[pos=0.25,left,lbl] {$W(0.2)$} (rrYX');
 \draw [tran,rounded corners] (rd.-40) |- node[pos=0.25,left,lbl] {$W(0.5)$} (rr);
 \draw [tran,rounded corners] (rd) -- node[pos=0.15,above,lbl] {$W(1)$} (rrYXrYX'r);
 \draw [tran,rounded corners] (rd.18) -- node[midway,above,lbl] {$W(0.5)$} (rrYXrYX');
 \draw [tran,rounded corners] (rd.-18) -- node[midway,above,lbl] {$W(0.8)$} (rrYXr);
 \draw [tran,rounded corners] (rd.30) |- node[pos=0.6,above,lbl] {$W(0.7)$} (rrYX'r);

 \draw [tran,rounded corners] (rrYX'.10) -- node[pos=0.1,below,lbl] {$\#$} (rYX'.165);
 \draw [tran,rounded corners] (rrYX) -- node[pos=0.1,below,lbl] {$\#$} +(1.3,0) -- (rYX.south west);
 \draw [tran,rounded corners] (rr.-15) -- node[pos=0.1,above,lbl] {$\#$} (r.-150);
 \draw [tran,rounded corners] (rrYXrYX'r.5) -- node[pos=0.7,above,lbl] {$\#$} +(0.5,0) -- (rYX'.south west);
 \draw [tran,rounded corners] (rrYXrYX'r) -- node[pos=0.6,below,lbl] {$\#$} (rYX);
 \draw [tran,rounded corners] (rrYXrYX'r.-5) -- node[pos=0.8,below,lbl] {$\#$} +(0.5,-1) -- (r.175);
 \draw [tran,rounded corners] (rrYX'r.-5) -- node[pos=0.13,above,lbl] {$\#$} +(0.75,-1.5) -- (r);
 \draw [tran,rounded corners] (rrYX'r) -- node[pos=0.1,above,lbl] {$\#$} (rYX');
 \draw [tran,rounded corners] (rrYXr.0) -- node[pos=0.03,above,lbl] {$\#$} +(2,0) -- (rYX);
 \draw [tran,rounded corners] (rrYXr.-6) -| node[pos=0.03,below,lbl] {$\#$} (r);
 \draw [tran,rounded corners] (rrYXrYX'.0) -- node[pos=0.03,below,lbl] {$\#$} +(2,0) -- (rYX);
 \draw [tran,rounded corners] (rrYXrYX'.6) -| node[pos=0.03,above,lbl] {$\#$} (rYX');

 \draw [thick, dotted] (rYX) -- +(.3,0);
 \draw [thick, dotted] (rYX') -- +(.3,0);

 \draw [dashed,->] (-0.1,1) -- node[midway,above,align=center] {pick/match an action\\ and a distribution} (1,1);
 \draw [dashed,->] (1.2,1) -- node[midway,above,align=center] {pick/match a subset} (2.2,1);
 \draw [dashed,->] (2.5,1) -- node[midway,above,align=center] {pick/match an element} (4.5,1);

 %\draw [dashed] (0.3,0.6) -- (0.3,\bpo-1.1);
 %\draw [dashed] (1.45,0.6) -- (1.45,\bpo-1.1);
 %\draw [dashed] (4.2,0.6) -- (4.2,\bpo-1.1);
\end{tikzpicture}
\end{center}
\caption{An example of the construction for Lemma~\ref{lem:nondet-constr-correct}. Here,
an arrow labelled $W(x)$ is an abbreviation for multiple transitions
labelled by all multiples of $0.1$ between $0.1$ and $x$.
\label{fig:reduction}}
\end{figure}
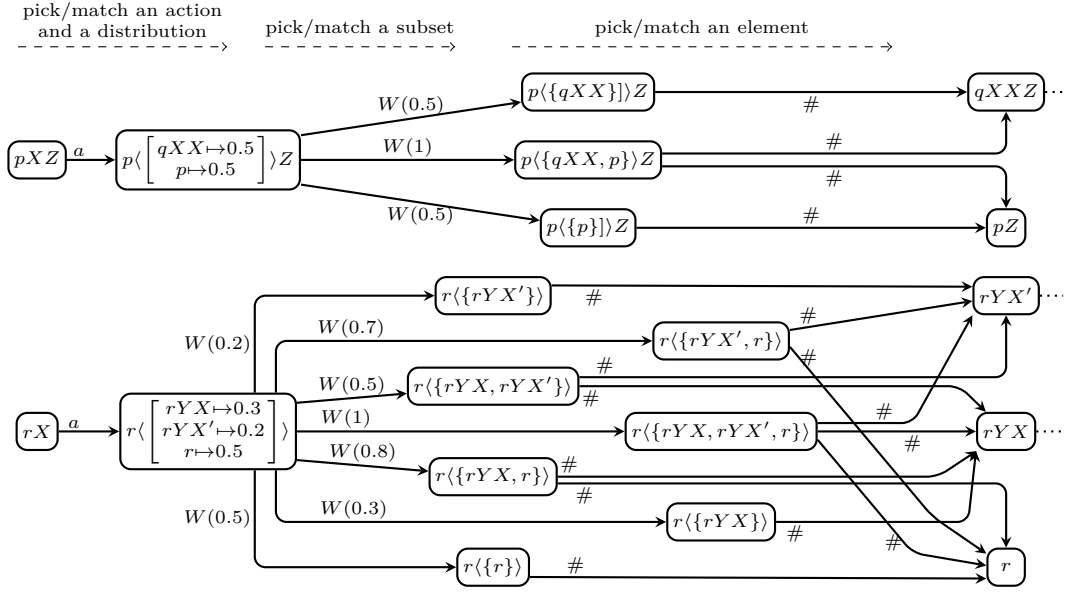

Lemma~\ref{lem:nondet-constr-correct} gives rise to the following theorem.
\begin{theorem}\label{thm:prob-to-nondet}
For any pPDA
$\Delta$ there is a PDA $\Delta'$ constructible in exponential time such that
for any configurations $q_1\gamma_1, q_2\gamma_2$ of $\Delta$ we have
$q_1\gamma_1 \sim q_2\gamma_2$ in $\Delta$ if and only if $q_1\gamma_1
\sim q_2\gamma_2$ in $\Delta'$.
In addition, if $\Delta$ is a pBPA, then
$\Delta'$ is a BPA.
\end{theorem}

Using Theorem~\ref{thm:prob-to-nondet} and \cite{Senizergues05,Burkart00}, we get the following corollary.

\begin{corollary}
The bisimilarity problem for pPDA is decidable, and the bisimilarity
problem for pBPA is decidable in triply exponential time.
\end{corollary}

\section{Upper Bounds} \label{sec-upper-bounds}

\subsection{Bisimilarity
%of Probabilistic One-Counter Automata} \label{sec:bispOCAinPSPACE}
of pOCA is in PSPACE}\label{sec:bispOCAinPSPACE}

The bisimilarity problem for (non-probabilistic)
one-counter automata is  PSPACE-complete, as shown
in~\cite{BGJ:Concur10}. It turns out that for pOCA we get
PSPACE-completeness as well.
The lower bound is shown in Section~\ref{sec-lower-bounds}; here
we show:

\begin{theorem} \label{thm-bisim-pOCA-inPspace}
The bisimilarity problem for pOCA is in PSPACE, even
if we present the instance
 $\Delta=(Q,\{Z,X\},\Sigma,\mathord{\btran{}})$,
$p X^mZ, q X^nZ$ (for which we ask if $p X^mZ\sim q X^nZ$)
by a shorthand using $m,n$ written in binary.
\end{theorem}
The reduction underlying
Theorem~\ref{thm:prob-to-nondet} would only provide an
exponential-space upper bound, so
we give a pOCA-specific polynomial-space algorithm.
%
%The \emph{bisimilarity problem for pOCAs}, Bis-pOCA,
%is the following: given a pOCA $\Delta$ and two
%configurations $p_0X^{m_0}$, $q_0X^{n_0}$, decide if  $p_0X^{m_0}\sim
%q_0X^{n_0}$.
%
%Subsection~\ref{subsec:inpspace} shows that   Bis-pOCA is in PSPACE,
%even for $m_0, n_0$ written in binary.
%%(and thus only $\log m_0, \log n_0$ contribute to the input
%%size).
%
%
In fact, we adapt the algorithm from~\cite{BGJ:Concur10};
the principles are the same but some
ingredients have to be slightly modified.
The following text is meant to give the idea in a self-contained
manner, though at a more abstract
level than in~\cite{BGJ:Concur10}.
%for a more detailed description, but some
%ingredients have to be modified and/or reproved.
%one counter automata (pOCA).
The main difference is in the notion of local consistency, discussed
around Proposition~\ref{prop:localconsistency}.

Similarly as~\cite{BGJ:Concur10},
we use a geometrical presentation of relations on
the set of configurations (Fig.~\ref{fig:belts} reflects such a
presentation).
A relation can be identified with
a $1$/$0$ (or YES/NO)
%, or black/white)
colouring of the ``grid''  $\N\times\N\times (Q\times Q)$:

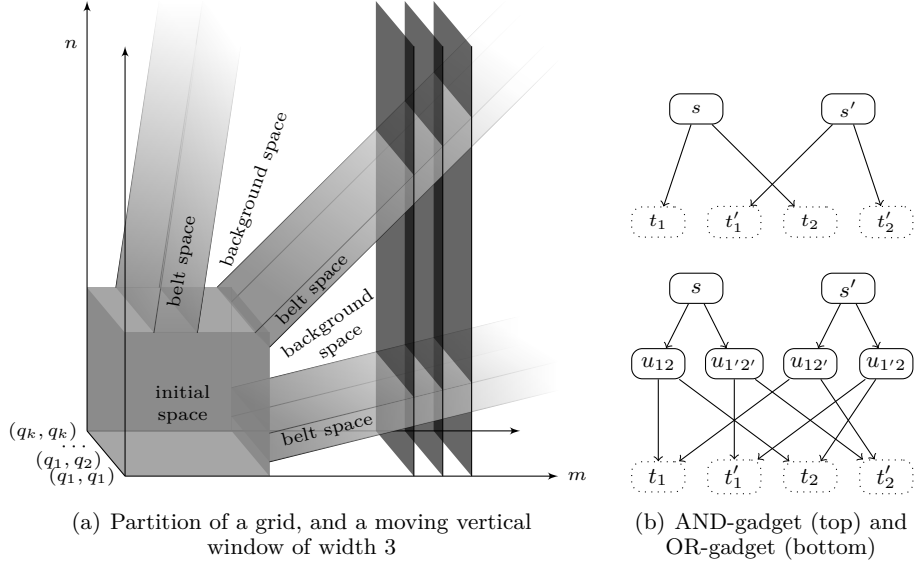
\begin{figure}[t]
\subfigure[Partition of a grid, and a moving vertical window~of~width~3\label{fig:belts}]{
\begin{tikzpicture}[x = {(1.9cm,0cm)}, y = {(0cm,1.9cm)}, z = {(-0.5cm,0.6cm)},font=\scriptsize]
%coords
 \path [name path=w1](2.0,0,0)--(2.0,3,0);
 \path [name path=w2](2.0,0,1)--(2.0,3,1);
 \path [name path=ww1](2.2,0,0)--(2.2,3,0);
 \path [name path=ww2](2.2,0,1)--(2.2,3,1);
 \path [name path=www1](2.4,0,0)--(2.4,3,0);
 \path [name path=www2](2.4,0,1)--(2.4,3,1);
 \path [name path=r1](1,0.3,0) -- (3,0.8,0);
 \path [name path=r2](1,0.3,1) -- (3,0.8,1);
 \path [name path=u1](0.9,1,0) -- (2.9,3,0);
 \path [name path=u2](0.9,1,1) -- (2.9,3,1);
 \path [name intersections={of= w1 and r1, name=A}];
 \path [name intersections={of= w2 and r2, name=B}];
 \path [name intersections={of= w1 and u1, name=C}];
 \path [name intersections={of= w2 and u2, name=D}];
 \path [name intersections={of= ww1 and r1, name=Aw}];
 \path [name intersections={of= ww2 and r2, name=Bw}];
 \path [name intersections={of= ww1 and u1, name=Cw}];
 \path [name intersections={of= ww2 and u2, name=Dw}];
 \path [name intersections={of= www1 and r1, name=Aww}];
 \path [name intersections={of= www2 and r2, name=Bww}];
 \path [name intersections={of= www1 and u1, name=Cww}];
 \path [name intersections={of= www2 and u2, name=Dww}];

%axe bg
 \draw[-latex'] (0,0,1) -- (3,0,1);

%belt right bg
 \draw[path fading=north] (1,0.1,1) -- (3,0.6,1);
 \draw[path fading=north] (1,0.3,1) -- (3,0.8,1);
 \shade[opacity=0.6, shading angle=90] (1,0.1,1) -- (1,0.3,1) -- (3,0.8,1) -- (3,0.6,1) -- cycle;
 \shade[opacity=0.6, shading angle=120] (1,0.1,0) -- (3,0.6,0) -- (3,0.6,1) -- (1,0.1,1) -- cycle;

%belt up right
 \draw[path fading=east] (0.9,1,1) -- (2.9,3,1);
 \draw[path fading=east] (1,0.9,1) -- (3,2.9,1);
 \shade[opacity=0.6, shading angle=125] (0.9,1,1) -- (1,0.9,1) -- (3,2.9,1) -- (2.9,3,1) -- cycle;
 \shade[opacity=0.6, shading angle=145] (1,0.9,0) -- (3,2.9,0) -- (3,2.9,1) -- (1,0.9,1) -- cycle;

%the box
 \path[fill=gray,opacity=0.7] (0,1,0) rectangle (1,0,0);
 \path[fill=gray,opacity=0.7] (0,1,1) rectangle (1,0,1);
 \path[fill=gray,opacity=0.7] (1,1,1) -- (1,1,0) -- (1,0,0) -- (1,0,1) -- cycle;
 \path[fill=gray,opacity=0.7] (0,1,1) -- (0,1,0) -- (0,0,0) -- (0,0,1) -- cycle;
 \node at (0.4,0.5,0) {\parbox{1.5cm}{\centering initial \\ space}};

%belt up
 \draw[path fading=north] (0.2,1,0) -- (0.5,3,0);
 \draw[path fading=north] (0.2,1,1) -- (0.5,3,1);
 \draw[path fading=north] (0.5,1,0) -- (0.8,3,0);
 \draw[path fading=north] (0.5,1,1) -- (0.8,3,1);
 \shade[opacity=0.6, shading angle=160] (0.2,1,0) -- (0.5,3,0) -- (0.5,3,1) -- (0.2,1,1) -- cycle;
 \shade[opacity=0.6, shading angle=180] (0.2,1,0) -- (0.5,1,0) -- (0.8,3,0) -- (0.5,3,0) -- cycle;
 \shade[opacity=0.6, shading angle=160] (0.5,1,0) -- (0.8,3,0) -- (0.8,3,1) -- (0.5,1,1) -- cycle;
 \shade[opacity=0.6, shading angle=180] (0.2,1,1) -- (0.5,1,1) -- (0.8,3,1) -- (0.5,3,1) -- cycle;
 \node[rotate=80] at (0.4,1.5,0) {belt space};

%windows (bottom parts
 \path[opacity=0.6,fill] (2.4,0,0) -- (Aww-1) -- (Bww-1) -- (2.4,0,1) -- cycle;
 \path[opacity=0.6,fill] (2.2,0,0) -- (Aw-1) -- (Bw-1) -- (2.2,0,1) -- cycle;
 \path[opacity=0.6,fill] (2,0,0) -- (A-1) -- (B-1) -- (2,0,1) -- cycle;

%belt right fg
 \draw[path fading=north] (1,0.1,0) -- (3,0.6,0);
 \draw[path fading=north] (1,0.3,0) -- (3,0.8,0);
 \shade[opacity=0.6, shading angle=120] (1,0.3,0) -- (3,0.8,0) -- (3,0.8,1) -- (1,0.3,1) -- cycle;
 %part of window here (middle)
 \path[opacity=0.6,fill] (Aww-1) -- (Cww-1) -- (Dww-1) -- (Bww-1) -- cycle;
 \path[opacity=0.6,fill] (Aw-1) -- (Cw-1) -- (Dw-1) -- (Bw-1) -- cycle;
 \path[opacity=0.6,fill] (A-1) -- (C-1) -- (D-1) -- (B-1) -- cycle;
 \shade[opacity=0.6, shading angle=90] (1,0.1,0) -- (1,0.3,0) -- (3,0.8,0) -- (3,0.6,0) -- cycle;
 \node[rotate=14] at (1.4,0.32,0) {belt space};

%belt right up fg
 \draw[path fading=east] (0.9,1,0) -- (2.9,3,0);
 \draw[path fading=east] (1,0.9,0) -- (3,2.9,0);
 \shade[opacity=0.6, shading angle=145] (0.9,1,0) -- (2.9,3,0) -- (2.9,3,1) -- (0.9,1,1) -- cycle;
 %part of window here (top)
 \path[opacity=0.6,fill] (2.4,3,0) -- (Cww-1) -- (Dww-1) -- (2.4,3,1) -- cycle;
 \path[opacity=0.6,fill] (2.2,3,0) -- (Cw-1) -- (Dw-1) -- (2.2,3,1) -- cycle;
 \path[opacity=0.6,fill] (2,3,0) -- (C-1) -- (D-1) -- (2,3,1) -- cycle;
 \shade[opacity=0.6, shading angle=145] (0.9,1,0) -- (1,0.9,0) -- (3,2.9,0) -- (2.9,3,0) -- cycle;
 \node[rotate=45] at (1.3,1.3,0) {belt space};

%axes fg
 \draw[-latex'] (0,0,1) -- (0,3,1) node[pos=0.9,left] {$n$};
 \draw[-latex'] (0,0,0) -- (3,0,0) node[pos=1,right] {$m$};
 \draw[-latex'] (0,0,0) -- (0,3,0);
 \draw (0,0,0) -- (0,0,1) node[pos=0.0,left] {$(q_1,q_1)$}
  node[pos=0.35,left] {$(q_1,q_2)$}
  node[pos=0.65,left] {$\cdots$}
  node[pos=1,left] {$(q_k,q_k)$};
%text
 \node[rotate=70] at (0.9,2,0) {background space};
 \node[rotate=35] at (1.45,1,0) {\parbox{1.5cm}{\centering background\\ space}};
%lines on windows
 \draw (2,0,0) -- (2,3,0);
 \draw (2.2,0,0) -- (2.2,3,0);
 \draw (2.4,0,0) -- (2.4,3,0);
\end{tikzpicture}
}
\subfigure[AND-gadget (top) and OR-gadget (bottom)\label{fig:gadgets}]{
\parbox[b]{3.8cm}{
\begin{tikzpicture}[font=\scriptsize]
\tikzstyle{every node} = [inner sep=2pt];
\tikzstyle{mystate} = [minimum height=4mm, minimum width=7mm,rounded corners];

%\node (aaa) at (1,2) [draw,circle] {$\wedge$};
%\node at (1,2) [above=7pt,right=7pt] {$q_i$};
%\node (bbb) at (0,0) [draw,dotted,circle] {\phantom{$v_i$}};
%\node at (0,0) [above=7pt,left=7pt] {$q_j$};
%\node (ccc) at (2,0) [draw,dotted,circle] {\phantom{$v_i$}};
%\node at (2,0) [above=7pt,right=7pt] {$q_k$};
%
%\draw[->] (aaa)--(bbb);
%\draw[->] (aaa)--(ccc);

\node (aa) at (0.5,1.5) [mystate,draw] {$s$};
\node (bb) at (2.5,1.5) [mystate,draw] {$s'$};
\node (cc) at (0,0) [mystate,draw,dotted] {$t_1$};
\node (dd) at (1,0) [mystate,draw,dotted] {$t_1'$};
\node (ee) at (2,0) [mystate,draw,dotted] {$t_2$};
\node (ff) at (3,0) [mystate,draw,dotted] {$t_2'$};

\draw[->] (aa)--(cc);
\draw[->] (bb)--(dd);
\draw[->] (aa)--(ee);
\draw[->] (bb)--(ff);
\end{tikzpicture}

%\begin{tikzpicture}
% \node[state] (alpha) {$s$};
% \node[state, right=2cm of alpha] (alphap) {$s'$};
% \node[state, below left=1cm and 0.2cm of alpha] (beta1) {$t_1$};
% \node[state, below left=1cm and 0.2cm of alphap] (beta1p) {$t'_1$};
% \node[state, below right=1cm and 0.2cm of alpha] (beta2) {$t_2$};
% \node[state, below right=1cm and 0.2cm of alphap] (beta2p) {$t'_2$};
%
% \node[distr,below=4mm of alpha] (alphad) {};
% \node[distr,below=4mm of alphap] (alphapd) {};
%
% \draw[arr] (alpha) -- (alphad) node[midway,left] {$a$} -- (beta1);
% \draw[arr] (alphad)--(beta2);
% \draw[arr] (alphap) -- (alphapd) node[midway,left] {$a$} --(beta1p);
% \draw[arr] (alphapd) -- (beta2p);
%\end{tikzpicture}

\bigskip
\noindent
\begin{tikzpicture}[font=\footnotesize]
\tikzstyle{every node} = [inner sep=2pt];
\tikzstyle{mystate} = [minimum height=4mm, minimum width=7mm,rounded corners];
%\node (aa) at (1,3) [draw,circle] {$\vee$};
%\node at (1,3) [above=7pt,right=7pt] {$n_i$};
%\node (bb) at (0,1) [draw,dotted,circle] {\phantom{$v_i$}};
%\node at (0,1) [above=7pt,left=7pt] {$n_j$};
%\node (cc) at (2,1) [draw,dotted,circle] {\phantom{$v_i$}};
%\node at (2,1) [above=7pt,right=7pt] {$n_k$};
%
%\draw[->] (bb)--(aa);
%\draw[->] (cc)--(aa);

\node (a)  at (0.5,3) [mystate,draw] {$s$};
\node (b) at (2.5,3) [mystate,draw] {$s'$};
\node (c) at  (0,2) [mystate,draw] {$u_{12}$};
\node (d) at  (1,2) [mystate,draw] {$u_{1'2'}$};
\node (e) at  (2,2) [mystate,draw] {$u_{12'}$};
\node (f) at  (3,2) [mystate,draw] {$u_{1'2}$};
\node (g) at  (0,0.5) [mystate,draw,dotted] {$t_1$};
\node (h) at  (1,0.5) [mystate,draw,dotted] {$t_1'$};
\node (i) at  (2,0.5) [mystate,draw,dotted] {$t_2$};
\node (j) at  (3,0.5) [mystate,draw,dotted] {$t_2'$};
\draw[->]  (a) -- (c);
\draw[->]  (a) -- (d);
\draw[->]  (b) -- (e);
\draw[->]  (b) -- (f);
\draw[->]  (c) -- (g);
\draw[->]  (d) -- (h);
\draw[->]  (e) -- (j);
\draw[->]  (f) -- (i);
\draw[->]  (e) -- (g);
\draw[->]  (f) -- (h);
\draw[->]  (c) -- (i);
\draw[->]  (d) -- (j);
\end{tikzpicture}

%\begin{tikzpicture}
% \node[state] (alpha) {$s$};
% \node[state, right=2cm of alpha] (alphap) {$s'$};
%
% \node[state, below left=1cm and 0.2cm of alpha] (gamma12) {$u_{12}$};
% \node[state, below left=1cm and 0.2cm of alphap] (gamma1p2) {$u_{1'2}$};
% \node[state, below right=1cm and 0.2cm of alphap] (gamma12p) {$u_{12'}$};
% \node[state, below right=1cm and 0.2cm of alpha] (gamma1p2p) {$u_{1'2'}$};
%
% \node[state, below=of gamma12] (beta1) {$t_1$};
% \node[state, below=of gamma1p2] (beta1p) {$t'_1$};
% \node[state, below=of gamma1p2p] (beta2) {$t_2$};
% \node[state, below=of gamma12p] (beta2p) {$t'_2$};
%
% \node[distr,below=4mm of alpha] (alphad) {};
% \node[distr,below=4mm of alphap] (alphapd) {};
% \node[distr,below=4mm of gamma12] (gamma12d) {};
% \node[distr,below=4mm of gamma1p2] (gamma1p2d) {};
% \node[distr,below=4mm of gamma12p] (gamma12pd) {};
% \node[distr,below=4mm of gamma1p2p] (gamma1p2pd) {};
%
% \draw[arr] (alpha) -- (alphad)  node[midway,left] {$a$} -- (gamma12);
% \draw[arr] (alphad) --  (gamma1p2p);
% \draw[arr] (alphap) -- (alphapd)  node[midway,left] {$a$} -- (gamma12p);
% \draw[arr] (alphapd)--(gamma1p2);
%
% \draw[arr] (gamma12)-- (gamma12d) node[midway,left] {$a$} -- (beta1);
% \draw[arr] (gamma12d)--(beta2);
% \draw[arr] (gamma1p2p)-- (gamma1p2pd) node[midway,left] {$a$} -- (beta1p);
% \draw[arr] (gamma1p2pd)--(beta2p);
% \draw[arr] (gamma12p)-- (gamma12pd) node[midway,left] {$a$} -- (beta1);
% \draw[arr] (gamma12pd)--(beta2p);
% \draw[arr] (gamma1p2)-- (gamma1p2d) node[midway,left] {$a$} -- (beta1p);
% \draw[arr] (gamma1p2d)--(beta2);
%\end{tikzpicture}
}}
\caption{Figures for Section \ref{sec:bispOCAinPSPACE} (left) and \ref{sec-lower-bounds} (right)}
\end{figure}

\begin{definition}
For a
%(general binary)
relation $R$ on $Q\times (\{X\}^*Z)$,
by the  \emph{(characteristic) colouring} $\chi_R$ we mean
the function
$\chi_R:\N\times\N\times (Q\times Q)\rightarrow \{1,0\}$
%\{\bullet,\circ\}$
where $\chi_R(m,n,(p,q))=1$ if and only if $(pX^mZ,qX^nZ)\in R$.
Given (a colouring)
$\chi:\N\times\N\times (Q\times Q)\rightarrow\{1,0\}$, by
$R_{\chi}$ we denote the relation
$R_{\chi}=\{(pX^mZ,qX^nZ)\mid \chi(m,n,(p,q))=1\}$.
\end{definition}
The algorithm uses the fact that $\chi_{\sim}$ is ``regular'', i.e.
%$\sim\, =
$\{(m,n,(p,q))\mid p X^mZ\sim qX^nZ\}$ is a (special) semilinear set.
More concretely, there are polynomials $pol_1, pol_2: \N\rightarrow\N$
(independent of the pOCA~$\Delta$) such that the following partition of the grid
$\N\times\N\times (Q\times Q)$
(sketched in Fig.~\ref{fig:belts})
has an important property specified later.
If $Q=\{q_1,q_2,\dots,q_k\}$, hence $|Q|=k$, then the grid is
partitioned into three parts: the \emph{initial-space}, i.e.
$\{(m,n,(p,q))\mid m,n\leq pol_2(k)\}$, the \emph{belt-space},
which is given by at most $k^4$ linear belts, with the slopes $\frac{c}{d}$ where
$c,d\in\{1,2,\dots,k^2\}$ and with the (vertical) thickness bounded by
$pol_1(k)$, and the rest, called the \emph{background}.
Moreover, $pol_2(k)$ is sufficiently large w.r.t. $pol_1(k)$, so that
 the belts are separated by the background
 outside
the initial space.

The mentioned important property is that
there is a period $\psi$, given by an exponential function of $k$,
such that
if two points $(m,n,(p,q))$ and $(m+i\psi,n+j\psi,(p,q))$
(for $i,j\in\N$)
are both in the background, for both $m,n$ larger then a polynomial bound, then
$\chi_{\sim}$ has the same value for both these points; in other words,
$\chi_{\sim}$ colours the background periodically.
Another important ingredient is the locality of the bisimulation conditions,
resulting from the fact that the counter value can change by at most $1$
per step.

%Before outlining the algorithm, we now sketch an explanation of the
%``grid-partition''.
%We start with
To explain the ``grid-partition'', we start with
 considering the finite automaton $\F_{\Delta}$
underlying $\Delta$;
$\F_{\Delta}$ behaves like $\Delta$ ``pretending'' that the
counter is always positive.

%by using $\F_{\Delta}$ we also easily establish
%a trivial polynomial-space procedure computing
%$\chi_{\sim}$ for the background points.

\begin{definition}\label{D underlying}
For a pOCA
$\Delta=(Q,\{Z,X\},\Sigma,\mathord{\btran{}})$,
in the {\em underlying finite pLTS}
$\F_\Delta=(Q,\Sigma,\tran{})$ we have
a transition $p\tran{a}d'$ if and only if there is a transition $pX\btran{a}d$
such that $d'(q)=d(q,\varepsilon)+d(q,X)+d(q,XX)$ (for all
$q\in Q$).
\end{definition}
%
%\noindent
%Hence $\F_\Delta$ behaves like $\Delta$ when counter values are positive.
%Figure~\ref{fig:underlyingFA} shows $\F_\Delta$ which we get from $\T(\Delta)$
%in Figure~\ref{fig:TSROCA}.
%
Using standard partition-refinement arguments, we observe
that $\sim_{k-1}=\sim_k=\sim$ on $\F_{\Delta}$ when $k=|Q|$.
For configurations of~$\Delta$ we now define the distance
$\distINC$ to the set of configurations which are
``INCompatible'' with $\F_{\Delta}$.
%(Here $\infty$ is used as an infinite number.)
%
\begin{definition}
Assuming a pOCA
$\Delta=(Q,\{Z,X\},\Sigma,\btran{})$, where $|Q|=k$,
\\
we define $\INC\subseteq Q\times (\{X\}^* Z)$ and
$\distINC:Q\times  (\{X\}^* Z) \rightarrow\N\cup\{\infty\}$ as follows:
%the set $\INC$ as those configurations
%of $\S(\Delta)$ that
%are incompatible (not $k$-equivalent)
%to {\em all} states in $\F_\Delta$, formally
\begin{itemize}
\item
$\INC=\{pX^mZ\mid \forall q\in Q:pX^mZ\not\sim_k  q\}$ (where $q$
is a state in $\F_{\Delta}$),
\item
%$\distINC(pX^mZ)=\min\,\{\,\ell\mid
%\exists w:
%|w|=\ell\land pX^mZ\tran{w}\INC\,\}$\,;
$\distINC(pX^mZ)=\min\,\{\,\ell\mid\exists q\gamma\in\INC:
pX^mZ(\btran{})^\ell q\gamma \,\}$\,;
we set $\min\,\emptyset\,=\infty$.
\end{itemize}
\end{definition}
Since $pX^mZ\sim_m p$ (by induction on $m$), and thus
$pX^mZ\in\INC$ implies
$m<k$,
we can surely construct $\INC$ for a given pOCA in polynomial space.
%is sufficient for us).

\begin{proposition}
\label{L nonreachable}
\hfill
\begin{enumerate}
\item
If $pX^mZ\sim qX^nZ$ then $\distINC(pX^mZ)=\distINC(qX^nZ)$.
\item
\mbox{If %$pX^mZ\not\tran{}^*\INC$, $qX^nZ\not\tran{}^*\INC$ (so
$\distINC(pX^mZ)=\distINC(qX^nZ)=\infty$
then
$pX^mZ\sim qX^nZ$ iff $pX^mZ\sim_k qX^nZ$.}
\end{enumerate}
\end{proposition}
The proof is the same as in the non-probabilistic case.
(Point 1 is obvious. For Point 2 we verify that the set
$\{\,(q_1X^{n_1}Z,q_2X^{n_2}Z)\mid q_1X^{n_1}Z\sim_k
q_2X^{n_2}Z\text{ and } \distINC(q_1X^{n_1}Z) = \distINC(q_2X^{n_2}Z) = \infty$
is a bisimulation.)

%Now the case $\distINC(pX^{m}Z)=\distINC(qX^{n}Z)<\infty$
%deserves a special attention.
%If we contemplate a bit how a shortest path from
%$pX^{m}Z$ to $\INC$ (for large $m$) can look like, it is not difficult
%to anticipate (and routinely prove as in the previous works)
%that such a path can be based on iterating
%a simple counter-decreasing cycle (of length
%$\leq k$), possibly preceded by a polynomial prefix and followed by a
%polynomial suffix.
%So (finite) $\distINC(pX^{m}Z)$ can be always expressed
%by the use of linear functions $\frac{\ell}{e}m+b$ where
%$\ell, e\leq k$ are the length and the decreasing effect of a simple
%cycle and $b$ is bounded by a polynomial in $k$.
%Hence $\distINC(pX^{m}Z)=\distINC(qX^{n}Z)<\infty$
%gives an equation $n=\frac{\ell_1e_2}{e_1\ell_2}m+b'$
%and shows that
%$(m,n,(p,q))$ lies in one of the above mentioned belts, or in the
%initial space when $m,n$ are small.
Consider a shortest path from $pX^{m}Z$ to $\INC$ (for large $m$).
It is not hard to prove (as in~\cite[Lemma 10]{BGJ:Concur10}) that such a path can be based on iterating
a simple counter-decreasing cycle (of length
$\leq k$), possibly preceded by a polynomial prefix and followed by a
polynomial suffix.
So (finite) $\distINC(pX^{m}Z)$ can be always expressed
by the use of linear functions $\frac{\ell}{e}m+b$ where
$\ell, e\leq k$ are the length and the decreasing effect of a simple
cycle and $b$ is bounded by a polynomial in $k$.
It follows that if we have $\distINC(pX^{m}Z)=\distINC(qX^{n}Z)<\infty$,
then $n=\frac{\ell_1e_2}{e_1\ell_2}m+b'$,
which shows that
$(m,n,(p,q))$ lies in one of the above mentioned belts, or in the
initial space when $m,n$ are small.

%Since in the background points $(m,n,(p,q))$
%we have either $\distINC(pX^{m}Z)=\distINC(qX^{n}Z)=\infty$,
%and $\chi_{\sim}(m,n,(p,q))=1$ iff $pX^{m}Z\sim_k qX^{n}Z$,
%or $\distINC(pX^{m}Z)\neq\distINC(qX^{n}Z)$
%(and thus  $\chi_{\sim}(m,n,(p,q))=0$),
%we can obviously compute $\chi_{\sim}$ for any background point in
%polynomial space.

As a consequence, in the background points $(m,n,(p,q))$
we have either $\distINC(pX^{m}Z)=\distINC(qX^{n}Z)=\infty$,
and $\chi_{\sim}(m,n,(p,q))=1$ if and only if $pX^{m}Z\sim_k qX^{n}Z$,
or $\distINC(pX^{m}Z)\neq\distINC(qX^{n}Z)$
(and thus  $\chi_{\sim}(m,n,(p,q))=0$).
So we can easily compute $\chi_{\sim}$ for any background point in
polynomial space.

The above mentioned shortest paths to $\INC$ also show that
if we choose $\psi=k!$ (so $\psi=O(2^{k\log k})$) then
we have
$pX^{m}Z\tran{}^* \INC$ if and only if $pX^{(m+\psi)}Z\tran{}^*\INC$ (for $m$ larger than
some polynomial bound), since the counter-effect of each simple cycle
divides $\psi$.
Hence
$\psi$ is a background period as mentioned above.

A nondeterministic algorithm, verifying that $p_0X^{m_0}Z\sim
q_0X^{n_0}Z$ for  $(m_0,n_0,(p_0,q_0))$ in the initial or belt-space,
is based on ``moving a
vertical window of width $3$'' (as depicted in Fig.~\ref{fig:belts});
in each phase, the window is moved by $1$ (to the right),
its intersection with the initial and belt space
(containing polynomially many
points) is computed, a colouring
on this intersection is guessed
 ($\chi_{\sim}$ is intended)
and its (local) consistency is
checked (for which also $\chi_{\sim}$
on the neighbouring background points is computed).
More precisely, in the first, i.e. leftmost, window position a
colouring in all
three (vertical) slices is guessed and the local consistency in the first two
slices is checked; after any later shift of the window
by one to the right, a colouring in
the new (the rightmost) slice is
guessed (the guesses in the previous two slices being remembered), and
the consistency in the current middle slice is checked.
If this is successfully performed
for exponentially many steps,
after $(m_0,n_0,(p_0,q_0))$ has been coloured with $1$, then
it is guaranteed that
the algorithm could successfully run forever; the pigeonhole principle
induces that each belt could be periodically coloured, with an
exponential period compatible with the period of the background-border
of the belt.
Such a successful run of the algorithm, exponential in time but
obviously only
polynomial in the required space, is thus a witness of $p_0X^{m_0}Z\sim
q_0X^{n_0}Z$. Since PSPACE=NPSPACE, we have thus sketched a proof
 of Theorem~\ref{thm-bisim-pOCA-inPspace}.
%it holds also in the case when $m_0, n_0$ are given in
%binary (and thus only $\log m_0, \log n_0$ contribute to the input size).

It remains to define precisely the consistency of a colouring,
guaranteeing that a successful run of the algorithm really witnesses
$p_0X^{m_0}Z\sim q_0X^{n_0}Z$. (As already mentioned, this is the main
change wrt~\cite{BGJ:Concur10}.)
%(This is more
%sophisticated than in the non-probabilistic case.)
We use the following particular variant of
characterizing (probabilistic) bisimilarity.
%
%\petr{adapt the Preliminaries !?}
%
Given a pLTS
$(S,\Sigma,\tran{})$,
we say that $(s,t)$ is \emph{consistent w.r.t.}
a relation $R$ on $S$ (not necessarily an equivalence)
if
for each $s\tran{a}d$ there is  $t\tran{a}d'$,
and conversely for each $t\tran{a}d'$ there is  $s\tran{a}d$,
such that $d, d'$ are $R'$-equivalent where
$R'$
is the least equivalence containing the set
$\{(s',t')\mid s\tran{}s', t\tran{}t', (s',t')\in R\}$.
%($s\tran{}s'$ means that $s'$ is a possible one-step successor of $s$.)
%Given an equivalence $Eq$ on $S$, we say
%that the transitions $s\tran{a}d$,  $t\tran{a}d'$ are
%\emph{equiv-intermatchable wrt} $Eq$ if for each equivalence class $C$ of
%$Eq$ we have $d(C)=d'(C)$ (where $d(C)=\sum_{s'\in C}d(s)$).
%
%Given a relation $R$ on $S$, we say
%that the transitions $s\tran{a}d$,  $t\tran{a}d'$ are
%\emph{intermatchable wrt} $R$ if they are
%equiv-intermatchable wrt the least equivalence containing the set
%$\{(s',t')\mid d(s')>0, d'(t')>0, (s',t')\in R\}$.
A \emph{relation} $R$ is \emph{consistent} if each $(s,t)\in R$ is consistent w.r.t.
$R$. The following proposition can be verified along the standard
lines.

\begin{proposition}\label{prop:localconsistency}
$\sim$ is consistent.
%\\
If $R$ is consistent then $R\subseteq \mathord{\sim}$.
\end{proposition}
%
%\petr{Remarks. Questions}
%Is it ok ?? We should give a proof, or this is something classical ??
%It seems that this specific form is useful here for pOCA, where
%$(m,n,(p,q))$ can be coloured differently then $(n,m,(q,p))$, etc.
%
Our algorithm can surely (locally) check the above defined
consistency of the
constructed $\chi$
(i.e. of $R_{\chi}$).
%in $\lra{m,n,(p,q)}$ is determined by
%the values of $\chi$ on
%$\neighbours(m,n,(p,q))$, which is what we needed for the algorithm.
%%=\{\lra{m',n',(p',q')}\mid |m'-m|\leq 1, |n'-n|\leq 1\}.$$

%In other words, for checking the conditions for $(m,n,(p,q))$ it suffices
%to rely on the properties of the set
%$$\neighbours(m,n,(p,q))=\{\lra{m',n',(p',q')}\mid
%|m'-m|\leq 1, |n'-n|\leq 1\}\,.$$
%The respective ``(bisimulation)
%consistency'' of a colouring can be thus locally checked, though

\subsection{Bisimilarity of pvPDA is in EXPTIME} \label{sec-upper-bounds-visibly}

It is shown in~\cite[Theorem 3.3]{SrbaVisiblyPDA:2009} that the bisimilarity problem for (non-probabilistic) vPDA is EXPTIME-complete.
We will show that the same holds for pvPDA.
First we show the upper bound:

\newcommand{\stmtthmupperboundsvisibly}{
The bisimilarity problem for pvPDA is in EXPTIME.
}
\begin{theorem} \label{thm-upper-bounds-visibly}
\stmtthmupperboundsvisibly
\end{theorem}
In~\cite{SrbaVisiblyPDA:2009} the upper bound is proved using a reduction to the model-checking problem
 for (non-visibly) PDA and the modal $\mu$-calculus.
The latter problem is in EXPTIME by~\cite{Walukiewicz2001}.
This reduction does not apply in the probabilistic case.
The reduction from Section~\ref{sec-prob-to-nondet} cannot be directly applied either,
 since it incurs an exponential blowup, yielding only a double-exponential algorithm if combined with the result of~\cite{Walukiewicz2001}.
Therefore we proceed as follows:
First we give a direct proof for (non-probabilistic) vPDA,
 i.e., we show via a new proof that the bisimilarity problem for vPDA is in EXPTIME.
Then we show that the reduction from Section~\ref{sec-prob-to-nondet} yields a non-probabilistic vPDA
 that is exponential only in a way that the new algorithm can be made run in single-exponential time:
The crucial observation is that the reduction replaces each step in the pvPDA by three steps in the (non-probabilistic) vPDA.
An exponential blowup occurs only in intermediate states of the new LTS.
Our algorithm allows to deal with those states in a special pre-processing phase.
\iftechrep{%
See Appendix~\ref{app-upper-bounds}
}{%
See~\cite{techreport}
}%
for details.

\section{Lower Bounds} \label{sec-lower-bounds}

In this section we show hardness results for pOCA and pvPDA. We
start by defining two gadgets, adapted from~\cite{CBW12}, that will be
used for both results.  The gadgets are pLTS that allow us to simulate AND and
OR gates using probabilistic bisimilarity. We depict the gadgets in
Figure~\ref{fig:gadgets}, where we assume that
all edges have probability $1/2$ and have the same label. The gadgets satisfy
the following propositions
(here
$s \tran{a}t_1\mid t_2$ is a shorthand for
%two transitions
%$s \tran{a,0.5}t_1$, $s \tran{a,0.5}t_2$.
%a transition
$s\tran{a}d$ where $d(t_1)=d(t_2)=0.5$).

%\begin{figure}
%%\begin{center}
%\begin{minipage}{0.45\textwidth}
%\input{imgs/and-gadget}
%\end{minipage}
%\begin{minipage}{0.45\textwidth}
%\input{imgs/or-gadget}
%\end{minipage}
%%\end{center}
%\caption{AND-gadget (left) and OR-gadget (right)}\label{fig:gadgets}
%\end{figure}

\begin{proposition}\label{prop:ANDgadget}
\textnormal{(AND-gadget)}
Suppose $s,s'$,
$t_1,t'_1$, $t_2,t'_2$ are states in a pLTS such that
 $t_1\not\sim t'_2$ and the only transitions
outgoing from $s,s'$ are
%\begin{center}
$s\tran{a}t_1\mid t_2$
and
$s'\tran{a}t'_1\mid t'_2$\,.
%\end{center}
Then
%\begin{center}
$s\sim s'$ if and only if $t_1\sim t'_1$ $\land$
$t_2\sim t'_2$.
%\end{center}
\end{proposition}

\begin{proposition}\label{prop:ORgadget}
\textnormal{(OR-gadget)}
Suppose $s,s'$,
$t_1,t'_1$, $t_2,t'_2$,
and $u_{12}$, $u_{1'2}$, $u_{12'}$, $u_{1'2'}$
are states in a pLTS.
Let the only transitions
outgoing from
$s,s',u_{12}, u_{1'2},u_{12'},u_{1'2'}$
be
\begin{center}
$s\tran{a}u_{12}\mid u_{1'2'}$\,,
%\\
$s'\tran{a}u_{12'}\mid u_{1'2}$\,,
\\
$u_{12}\tran{a}t_1\mid t_2$\,,
$u_{1'2'}\tran{a}t'_1\mid t'_2$\,,
$u_{12'}\tran{a}t_1\mid t'_2$\,,
$u_{1'2}\tran{a}t'_1\mid t_2$\,.
\end{center}
Then
%\begin{center}
$s\sim s'$ if and only if $t_1\sim t'_1$ $\lor$
$t_2\sim t'_2$.
%\end{center}
\end{proposition}

\subsection{Bisimilarity of pOCA is PSPACE-hard}

In this section we prove the following:
\begin{theorem} \label{thm-pOCA-hard}
 Bisimilarity for pOCA is PSPACE-hard,
  even for unary (i.e., with only one action) and fully probabilistic pOCA,
   and for fixed initial configurations of the form $pXZ, qXZ$.
\end{theorem}
In combination with Theorem~\ref{thm-bisim-pOCA-inPspace} we obtain:
\begin{corollary}
 The bisimilarity problem for pOCA is PSPACE-complete.
\end{corollary}
\begin{proof}[Proof of Theorem~\ref{thm-pOCA-hard}]
We use a reduction from the emptiness problem for alternating
finite automata with a one-letter alphabet, known to be
PSPACE-complete~\cite{Holzer95,JancarSawaAFA:2007};
our reduction resembles the reduction in~\cite{SrbaVisiblyPDA:2009}
for (non-probabilistic) visibly one-counter automata.

A \emph{one-letter alphabet alternating finite automaton}, 1L-AFA, is a tuple
$A=(Q,\delta,q_0,F)$ where
$Q$ is the (finite) set of \emph{states},
%$\Sigma$ is a finite \emph{alphabet},
 $q_0$
is the \emph{initial state}, $F\subseteq Q$ is the set of
\emph{accepting states}, and the \emph{transition
function} $\delta$ assigns to each $q\in Q$
either $q_1\land q_2$ or $q_1\lor q_2$, where $q_1,q_2\in Q$.

We define the predicate $\Acc\subseteq Q\times\N$
by induction on the second component (i.e. the length of a one-letter word);
$\Acc(q,n)$ means
``$A$ starting in $q$ accepts $n$'':
%
%\begin{itemize}
%\item
$\Acc(q,0)$ if and only if $q\in F$;
%\item
$\Acc(q,n{+}1)$ if and only if
either
%$\delta(q)=q'$ and $\Acc(q',n)$, or
$\delta(q)=q_1\land q_2$ and
we have both $\Acc(q_1,n)$ and $\Acc(q_2,n)$, or
 $\delta(q)=q_1\lor q_2$ and we have
$\Acc(q_1,n)$ or $\Acc(q_2,n)$.
%\end{itemize}

The \emph{emptiness problem for 1L-AFA} asks, given a 1L-AFA $A$, if the set
$\{n\mid \Acc(q_0,n)\}$ is empty.

We reduce the emptiness of 1L-AFA to our problem.
We thus assume a 1L-AFA $(Q,\delta,q_0,F)$, and we construct a pOCA
$\Delta$ as
follows. $\Delta$ has $2|Q|+3$ `basic' states; the set of basic states
is $\{p,p',r\}\cup Q\cup Q'$ where
$Q'=\{q'\mid q\in Q\}$ is a copy of $Q$ and $r$ is a special dead
state.
Additional auxiliary states will be added to implement AND- and
OR-gadgets. $\Delta$ will have only one input
letter, denoted $a$, and will be fully probabilistic.

We aim to achieve $pXZ\sim p'XZ$ if and only if $\{n\mid \Acc(q_0,n)\}$ is
empty; another property will be that
\begin{equation}\label{eq:afapoca}
qX^{n}Z\sim q'X^{n}Z \textnormal{ if and only if } \neg\Acc(q,n).
\end{equation}
For each $q\in F$ we add a transition
$qZ \btran{a} d$ where $d(r,Z)=1$,
but $qZ$ is dead (i.e., there is no transition $qZ \btran{a}..$) if
$q\not\in F$;  $q'Z$ is dead for any $q'\in Q'$.
Both $rX$ and $rZ$ are dead as well.
Hence
(\ref{eq:afapoca}) is
satisfied for $n=0$.
Now we show (\ref{eq:afapoca}) holds for
$n>0$.

For $q$ with $\delta(q)=q_1\lor q_2$ we implement an
AND-gadget from Figure~\ref{fig:gadgets} (top) guaranteeing $qX^{n{+}1}Z\sim q'X^{n{+}1}Z$ if and only if
$q_1X^{n}Z\sim q_1'X^{n}Z$ and $q_2X^{n}Z\sim q_2'X^{n}Z$
(since $\neg\Acc(q,n{+}1)$ if and only if $\neg\Acc(q_1,n)$ and
$\neg\Acc(q_2,n)$):

We add rules
$qX \tran{} r_1X \mid r_2X$ (this is a shorthand for
$qX \btran{a} [r_1X \mapsto 0.5, r_2X\mapsto 0.5]$)
and $q'X \tran{} r'_1X \mid r'_2X$,
\\
and also $r_1X \tran{} q_1 \mid s_1X$,
$r_2X \tran{} q_2 \mid s_2X$,
 $r'_1X \tran{} q'_1 \mid s_1X$,
$r'_2X \tran{} q'_2 \mid s_2X$,
\\
and $s_1X\tran{0.5}s_1X$, $s_1X\tran{0.5}r$,
$s_2X\tran{0.4}s_2X$, $s_2X\tran{0.6}r$.
The intermediate states $r_1,r_2,r'_1,r'_2$, and $s_1,s_2$ serve to
implement the condition
$t_1\not\sim t'_2$ from Proposition~\ref{prop:ANDgadget}.

For $q$ with $\delta(q)=q_1\land q_2$ we (easily) implement an
OR-gadget from Figure~\ref{fig:gadgets} (bottom)
guaranteeing $qX^{n{+}1}Z\sim q'X^{n{+}1}Z$ if and only if
$q_1X^{n}Z\sim q_1'X^{n}Z$ or $q_2X^{n}Z\sim q_2'X^{n}Z$
(since $\neg\Acc(q,n{+}1)$ if and only if $\neg\Acc(q_1,n)$ or
$\neg\Acc(q_2,n)$).

To finish the construction, we add transitions
$pX\btran{a}d$ where $d(p,X^2)=d(q_0,\varepsilon)=d(r,X)=\frac{1}{3}$ and
$p'X\btran{a}d'$ where $d'(p',X^2)=d(q'_0,\varepsilon)=d(r,X)=\frac{1}{3}$;
the transitions added before guarantee that
$pX^{n+2}Z\not\sim q'_0X^nZ$ and
$q_0X^{n}Z\not\sim p'X^{n+2}Z$.
\end{proof}

%For the case of pBPA, we have the following theorem.
%\newcommand{\stmtthmbisimBPAPspacehard}{
%Bisimilarity for pBPA is PSPACE-hard,
%  even for unary (i.e., with only one input letter) and fully probabilistic pBPA.
%}
%\begin{theorem} \label{thm-bisim-BPA-Pspacehard}
% \stmtthmbisimBPAPspacehard
%\end{theorem}
%%
%For the proof we give a reduction from QBF (validity of fully Quantified Boolean Formulas).
%The reduction follows
%the one given by Srba~\cite{SrbaBPAandBPPhardness:2003} for non-probabilistic BPA,
%reusing the gadgets defined earlier.
%The details are in Appendix~\ref{app-bisim-BPA-Pspacehard}.

\subsection{Bisimilarity of pvPDA is EXPTIME-hard}

In this section we prove the following:
\begin{theorem} \label{thm-pvPDA-hard}
 Bisimilarity for pvPDA is EXPTIME-hard,
  even for fully probabilistic pvPDA with $|\Sigmar| = |\Sigmai| = |\Sigmac| = 1$.
%   and for fixed initial configurations of the form $pX, qY$.
\end{theorem}
In combination with Theorem~\ref{thm-upper-bounds-visibly} we obtain:
\begin{corollary}
 The bisimilarity problem for pvPDA is EXPTIME-complete.
\end{corollary}
It was shown in~\cite{SrbaVisiblyPDA:2009} that bisimilarity for (non-probabilistic) vPDA is EXPTIME-complete.
The hardness result there follows by observing that the proof given in~\cite{KuceraM10} for general PDA works in fact even for vPDA.
Referring to the conference version of~\cite{KuceraM10}, 
 it is commented in~\cite{SrbaVisiblyPDA:2009}: ``Though conceptually elegant, the technical details of the reduction are rather tedious.''
For those reasons we give a full reduction from the problem of determining the winner in a reachability game on pushdown processes.
This problem was shown EXPTIME-complete in~\cite{Walukiewicz2001}.
Our reduction proves Theorem~\ref{thm-pvPDA-hard}, i.e., for unary and fully probabilistic pvPDA,
 and at the same time provides a concise proof for (non-probabilistic) vPDA.
\begin{proof}[Proof of Theorem~\ref{thm-pvPDA-hard}]
Let $\Delta = (Q,\Gamma,\{a\},\mathord{\btran{}})$ be a unary non-probabilistic PDA with
 a control state partition $Q = Q_0 \cup Q_1$ and an initial configuration $p_0 X_0$.
We call a configuration $p X \alpha$ \emph{dead} if it has no successor configuration,
 i.e., if $\Delta$ does not have a rule with $p X$ on the left-hand side.
Consider the following game between Player~0 and Player~1 on the LTS~$\mathcal{S}(\Delta)$ induced by~$\Delta$:
The game starts in~$p_0 X_0$.
Whenever the game is in a configuration $p \alpha$ with $p \in Q_i$ (where $i \in \{0,1\}$),
 Player~$i$ chooses a successor configuration of~$p \alpha$ in~$\mathcal{S}(\Delta)$.
The goal of Player~$1$ is to reach a dead configuration; the goal of Player~$0$ is to avoid that.
It is shown in~\cite[pp.~261--262]{Walukiewicz2001} that determining the winner in that game is EXPTIME-hard.

W.l.o.g.\ we can assume that each configuration has at most two successor configurations,
 and that no configuration with empty stack is reachable.
% and that if a configuration has two successor configurations then they do not arise from pop or push transitions.
We construct a fully probabilistic pvPDA $\bar \Delta = (\bar Q, \Gamma, \{\ar, \ai, \ac\}, \mathord{\ctran{}})$
 such that the configurations $p_0 X_0$ and~$p_0' X_0$ of~$\bar \Delta$ are bisimilar if and only if Player~$0$ can win the game.
For each control state $p \in Q$ the set~$\bar Q$ includes $p$ and a copy~$p'$.

For each $p X \in Q \times \Gamma$, if $p X$ is dead in~$\Delta$, we add a rule $p X \ctran{\ai,1} p X$ in~$\bar \Delta$,
 and a rule $p' X \ctran{\ai,1} z X$ where $z \in \bar Q$ is a special control state not occurring on any left-hand side.
This ensures that if $p X$ is dead in~$\Delta$ (and hence Player~$1$ wins), then we have $p X \not\sim p' X$ in~$\bar \Delta$.

For each $p X \in Q \times \Gamma$ that has in~$\Delta$ a single successor configuration~$q \alpha$,
 we add rules $p X \ctran{a,1} q \alpha$ and $p' X \ctran{a,1} q' \alpha$,
 where $a = \ar, \ai, \ac$ if $|\alpha| = 0, 1, 2$, respectively.

For each $p X \in Q \times \Gamma$ that has in~$\Delta$ two successor configurations,
 let $p_1 \alpha_1$ and $p_2 \alpha_2$ denote the successor configurations.
W.l.o.g.\ we can assume that $\alpha_1 = X_1 \in \Gamma$ and $\alpha_2 = X_2 \in \Gamma$.
\begin{itemize}
 \item If $p \in Q_0$ we implement an OR-gadget from Figure~\ref{fig:gadgets}:
  let $(p_1X_1p_2X_2), (p_1'X_1p_2'X_2), (p_1X_1p_2'X_2), (p_1'X_1p_2X_2) \in \bar Q$ be fresh control states,
  and add rules $p X \ctran{} (p_1X_1p_2X_2) X \mid (p_1'X_1p_2'X_2) X$
   (this is a shorthand for $p X \ctran{\ai,0.5} (p_1X_1p_2X_2) X$ and $p X \ctran{\ai,0.5} (p_1'X_1p_2'X_2) X$)
            and  $p' X \ctran{} (p_1X_1p_2'X_2) X \mid (p_1'X_1p_2X_2) X$
  as well as
   $(p_1X_1p_2X_2) X \ctran{} p_1 X_1 \mid p_2 X_2$ and
   $(p_1'X_1p_2'X_2) X \ctran{} p_1' X_1 \mid p_2' X_2$ and
   $(p_1X_1p_2'X_2) X \ctran{} p_1 X_1 \mid p_2' X_2$ and
   $(p_1'X_1p_2X_2) X \ctran{} p_1' X_1 \mid p_2 X_2$.
 \item If $p \in Q_0$ we implement an AND-gadget from Figure~\ref{fig:gadgets}:
  let $(p_1X_1), (p_1'X_1), (p_2X_2), (p_2'X_2) \in \bar Q$ be fresh control states,
  and add rules $p X \ctran{} (p_1X_1) X \mid (p_2X_2) X$
            and $p' X \ctran{} (p_1'X_1) X \mid (p_2'X_2) X$
  as well as
   $(p_1X_1) X \ctran{\ai,1} p_1 X_1$ and
   $(p_1'X_1) X \ctran{\ai,1} p_1' X_1$ and
   $(p_2X_2) X \ctran{} p_2 X_2 \mid z X$ and
   $(p_2'X_2) X \ctran{} p_2' X_2 \mid z X$.
  Here, the transitions to~$z X$ serve to implement the condition $t_1\not\sim t'_2$ from Proposition~\ref{prop:ANDgadget}.
\end{itemize}
An induction argument now easily establishes that
 $p_0 X_0 \sim p_0' X_0$ holds in~$\bar\Delta$ if and only if Player~$0$ can win the game in~$\Delta$.

We remark that exactly the same reduction works for non-probabilistic vPDA,
 if the probabilistic branching is replaced by nondeterministic branching.
\end{proof}

\bigskip
\noindent
{\bf Acknowledgements.\quad}
The authors thank anonymous referees for their helpful feedback.
Vojt\v{e}ch Forejt is supported by a Newton International Fellowship of the Royal Society.
Petr Jan\v{c}ar is supported by the Grant Agency of the Czech Rep.
(project GA\v{C}R:P202/11/0340); his short visit at Oxford was also
supported by ESF-GAMES grant no. 4513.
Stefan Kiefer is supported by the EPSRC.

%\medskip
%\noindent \emph{Acknowledgements.}
%---------------------
\bibliographystyle{plain} %oder alpha oder splncs
\bibliography{db}

\begin{thebibliography}{10}

\bibitem{Baier96}
C.~Baier.
\newblock Polynomial time algorithms for testing probabilistic bisimulation and
  simulation.
\newblock In {\em CAV}, pages 50--61, 1996.

\bibitem{BGJ:Concur10}
S.~B{\"o}hm, S.~G{\"o}ller, and P.~Jan{\v{c}}ar.
\newblock Bisimilarity of one-counter processes is {PSPACE}-complete.
\newblock In {\em CONCUR}, volume 6269 of {\em LNCS}, pages 177--191, 2010.

\bibitem{Brazdil08}
T.~Br{\'a}zdil, A.~Ku\v{c}era, and O.~Stra\v{z}ovsk{\'y}.
\newblock Deciding probabilistic bisimilarity over infinite-state probabilistic
  systems.
\newblock {\em Acta Inf.}, 45(2):131--154, 2008.

\bibitem{Burkart00}
O.~Burkart, D.~Caucal, F.~Moller, and B.~Steffen.
\newblock Verification on infinite structures.
\newblock In J.A. Bergstra, A.~Ponse, and S.A. Smolka, editors, {\em Handbook
  of Process Algebra}, pages 545--623. North-Holland, 2001.

\bibitem{CBW12}
D.~Chen, F.~van Breugel, and J.~Worrell.
\newblock On the complexity of computing probabilistic bisimilarity.
\newblock In {\em FoSSaCS}, volume 7213 of {\em LNCS}, pages 437--451, 2012.

\bibitem{EWY10}
K.~Etessami, D.~Wojtczak, and M.~Yannakakis.
\newblock Quasi-birth-death processes, tree-like {QBDs}, probabilistic
  1-counter automata, and pushdown systems.
\newblock {\em Perform. Eval.}, 67(9):837--857, 2010.

\bibitem{FuKatoen11}
H.~Fu and J.-P. Katoen.
\newblock Deciding probabilistic simulation between probabilistic pushdown
  automata and finite-state systems.
\newblock In {\em FSTTCS}, pages 445--456, 2011.

\bibitem{Holzer95}
M.~Holzer.
\newblock On emptiness and counting for alternating finite automata.
\newblock In {\em Developments in Language Theory}, pages 88--97, 1995.

\bibitem{JancarSawaAFA:2007}
P.~Jan{\v{c}}ar and Z.~Sawa.
\newblock A note on emptiness for alternating finite automata with a one-letter
  alphabet.
\newblock {\em Inf. Process. Lett.}, 104(5):164--167, 2007.

\bibitem{Jancar12}
P.~Jan\v{c}ar.
\newblock Bisimilarity on {B}asic {P}rocess {A}lgebra is in 2-{E}xp{T}ime (an
  explicit proof).
\newblock {\em CoRR}, abs/1207.2479, 2012.

\bibitem{Kiefer12}
S.~Kiefer.
\newblock {BPA} bisimilarity is {EXPTIME}-hard.
\newblock {\em CoRR}, abs/1205.7041, 2012.

\bibitem{KuceraM10}
A.~Ku\v{c}era and R.~Mayr.
\newblock On the complexity of checking semantic equivalences between pushdown
  processes and finite-state processes.
\newblock {\em Information and Computation}, 208(7):772--796, 2010.

\bibitem{SL94}
R.~Segala and N.~A. Lynch.
\newblock Probabilistic simulations for probabilistic processes.
\newblock In {\em CONCUR}, volume 836 of {\em LNCS}, pages 481--496. Springer,
  1994.

\bibitem{Senizergues05}
G.~S{\'e}nizergues.
\newblock The bisimulation problem for equational graphs of finite out-degree.
\newblock {\em SIAM J. Comput.}, 34(5):1025--1106, 2005.

\bibitem{SrbaVisiblyPDA:2009}
J.~Srba.
\newblock Beyond language equivalence on visibly pushdown automata.
\newblock {\em Logical Methods in Computer Science}, 5(1):2, 2009.

\bibitem{Walukiewicz2001}
I.~Walukiewicz.
\newblock Pushdown processes: Games and model-checking.
\newblock {\em Information and Computation}, 164(2):234--263, 2001.

\end{thebibliography}

\iftechrep{
\newpage
\appendix
\section{Proofs omitted from Section~\ref{sec-prob-to-nondet}}
In this section we present proofs of some claims from Section~\ref{sec-prob-to-nondet}

\subsection{Proof of Lemma~\ref{lem:nondet-constr-correct}}\label{app:nondet-constr-correct}
Lemma~\ref{lem:nondet-constr-correct} follows immediately from
the following lemma.
\begin{lemma}
 For all configurations $qX\gamma$ and $rY\delta$ of $\Delta$ we have $qX\gamma\sim_n rY\delta$ in $\Delta$ if and only if $qX\gamma\sim'_{3\cdot n} rY\delta$ in $\Delta'$.
\end{lemma}
\begin{proof}
In what follows, given a distribution $d=[\alpha_1 \mapsto x_1, \ldots, \alpha_n\mapsto x_n]$, we use
$d^w$ to denote the distribution $[\alpha_1w \mapsto x_1, \ldots, \alpha_nw\mapsto x_n]$. Also, we use $\sim'$ to denote the relation $\sim$ of $\Delta'$,
to distinguish it from the relation $\sim$ of $\Delta$.

Let us start with the direction $\Rightarrow$ of the lemma. For $n=0$
the claim obviously holds. Assume it holds for all numbers lower than
$n$.  Let $qX\gamma$ and $rY\delta$ be configurations of $\Delta$ such
that $qX\gamma\sim_n rY\delta$.  W.l.o.g. let us pick any transition
$qX \ctran{a} q\stacksymba{d_1}$ where $d_1=[q_1\beta_1 \mapsto x_1,
  \ldots, q_n\beta_n\mapsto x_n]$.  There must be a transition
$rY\ctran{a} r\stacksymba{d_2}$ where $d_2=[r_1\beta_1 \mapsto y_1,
  \ldots, r_m\beta_m\mapsto y_m]$ such that $d_1^\gamma$ and
$d_2^\delta$ are $\sim_{n-1}$-equivalent. Let $q\stacksymba{d_1}
\ctran{x} q\stacksymba{\{q_{i_1}\alpha_{i_1}, \ldots
  q_{i_k}\alpha_{i_k}\}}$ be an arbitrary rule with $q\stacksymba{d_1}$
on the left hand side (the case of $r\stacksymba{d_2}$ is similar).
For the set $T=\{q_{i_1}\alpha_{i_1}\gamma,\ldots
q_{i_k}\alpha_{i_k}\gamma\}$, we have $x \le d_1(T)$ and there must be
a set
$T'=\{r_{j_1}\beta_{j_1}\delta,\ldots,r_{j_\ell}\beta_{j_\ell}\delta\}$
satisfying the conditions of Lemma~\ref{lem:bisim-as-game} such that
$d_2(T') \ge d_1(T) \ge x$. Hence there is an action
$r\stacksymba{d_2} \ctran{x} r\stacksymba{\{r_{j_1}\beta_{j_1}, \ldots
  r_{j_\ell}\beta_{j_\ell}\}}$. Because $T$ and $T'$ were chosen to
satisfy the conditions of Lemma~\ref{lem:bisim-as-game}, for any
$r\stacksymba{\{r_{j_1}\beta_{j_1}, \ldots
  r_{j_\ell}\beta_{j_\ell}\}}\delta\tran{b}r_j\beta_j \delta$ (these are
the only actions available) there is an action
$q\stacksymba{\{q_{i_1}\alpha_{i_1}, \ldots q_{i_k}\alpha_{i_k}\}}
\tran{b}q_i\alpha_i \gamma$ such that $r_j\beta_j \delta\sim_n
q_i\alpha_i \gamma$, and vice versa.

Now let us analyse $\Leftarrow$. For $n=0$ the claim obviously holds. Assume it holds for all numbers lower than $n$.
Let $qX\gamma$ and $rY\delta$ be configurations of $\Delta$ such that $qX\gamma\sim'_{3\cdot n} rY\delta$ in $\Delta'$. Let $qX\gamma\btran{a}d$ be arbitrary
rule, then there is a transition $qX\gamma\tran{a} q\stacksymba{d}\gamma$ in $\Delta'$ and  $rY\delta\tran{a} r\stacksymba{d'}\delta$
such that $q\stacksymba{d}\gamma \sim'_{3\cdot n -1} r\stacksymba{d'}\delta$. There is also a rule $rY\delta\btran{a}d'$, so to finish
the proof it suffices to see that $d^\gamma$ and $(d')^\delta$ are $\sim_{n-1}$-equivalent.
Let $T=\{q_{i_1}\alpha_{i_1}\gamma, \ldots, q_{i_k}\alpha_{i_k}\gamma\} \subseteq \support{d^\gamma}$ be arbitrary
(for the subsets of $\support{(d')^\gamma}$ the proof is similar), and let $x=d^\gamma(T)$. There is a transition
$q\stacksymba{d}\gamma \tran{x} q\stacksymba{\{q_{i_1}\alpha_{i_1}, \ldots, q_{i_k}\alpha_{i_k}\}}\gamma$, and hence a transition
$r\stacksymba{d'}\delta \tran{x} r\stacksymba{\{r_{j_1}\beta_{j_1}, \ldots, r_{j_\ell}\beta_{j_k}\}}\delta$ such that
\begin{equation}\label{eq:bisim}
q\stacksymba{\{q_{i_1}\alpha_{i_1}, \ldots, q_{i_k}\alpha_{i_k}\}}\gamma
\sim'_{3\cdot n - 2}
r\stacksymba{\{r_{j_1}\beta_{j_1}, \ldots, r_{j_\ell}\beta_{j_k}\}}\delta
\end{equation}
We put $T'=\{r_{j_1}\beta_{j_1}\delta, \ldots, r_{j_\ell}\beta_{j_\ell}\delta\}$. We show that $T$ and $T'$ satisfy the conditions
from Lemma~\ref{lem:bisim-as-game} for the relation $\sim_{n-1}$.
First, due to the construction of rules available under $x$ we have $d^\gamma(T) \le (d')^\delta(T')$.
Further, for an arbitrary element $q_{i}\alpha_{i}\gamma$ there is
a transition
$q\stacksymba{\{q_{i_1}\alpha_{i_1}, \ldots, q_{i_k}\alpha_{i_k}\}}\gamma \tran{b} q_{i}\alpha_{i}\gamma$,
and so due to Equation~\ref{eq:bisim} there must be a transition
$r\stacksymba{\{r_{j_1}\beta_{j_1}, \ldots, r_{j_\ell}\beta_{j_k}\}}\delta \tran{b} r_j\beta_j \delta$
such that $q_{i}\alpha_{i}\gamma \sim'_{3n-3} r_j\beta_j \delta$. Using the induction hypothesis we get $q_{i}\alpha_{i}\gamma \sim_{n-1} r_j\beta_j \delta$,
which finishes the proof.
\end{proof}
\subsection{Analysis of the size of $\Delta'$}\label{app:analysis-complexity}

Let us analyse the size of $\Delta'$. Let $|\varrho|$ be the
number of rules of $\Delta$,
and let $m$ be the maximal size
of the support of a distribution assigned by some rule of $\Delta$.
The size of $|\Gamma'|$
is at most $|\Gamma| + |\varrho| + |\varrho|\cdot 2^m$, the size of $\Sigma'$ is at most
$|\Sigma| + |R| + 1$, and the number of rules (in $\ctran{}$) under an action $a\in\Sigma$
is at most $|\varrho|$, under an action $x\in W$ it is at most $|\varrho| \cdot 2^m$, where $|W| \le |\varrho|\cdot 2^m$.
The number of rules under the action $\#$ is at most $|\varrho|\cdot 2^m \cdot m$. Hence the
size of $\Delta'$ is exponential in the size of $\Delta$, but polynomial when the size of the support
of distributions assigned by rules is fixed. Obviously, the construction can be done in time exponential (or polynomial, respectively)
in the size of $\Delta$.

\section{Proofs Omitted from Section~\ref{sec-upper-bounds}} \label{app-upper-bounds}

\subsection{Proof of Theorem~\ref{thm-upper-bounds-visibly}}

We prove the following theorem from the main body of the paper:

\begin{qtheorem}{\ref{thm-upper-bounds-visibly}}
\stmtthmupperboundsvisibly
\end{qtheorem}

\begin{proof}
The proof is structured as follows.  First we show that bisimilarity
for (non-probabilistic) vPDA is in EXPTIME, thus reproving a result
from~\cite{SrbaVisiblyPDA:2009} via a different method.  Then we show
that, although the reduction from Section~\ref{sec-prob-to-nondet}
yields an exponential blow-up in translating from pvPDA to vPDA, our
new algorithm for deciding bisimilarity on vPDA can still be made to
run in single-exponential time in the size of the original pvPDA.

Let $p_0 \alpha_0$ and $q_0 \beta_0$ be the given initial configurations.
W.l.o.g.\ we assume that $\alpha_0 = X_0 \in \Gamma$ and $\beta_0 = Y_0 \beta'$ with $Y_0 \in \Gamma$ and $\beta' \in \Gamma^*$.
Recall that bisimulation in a labelled transition system can be naturally characterised by a \emph{bisimulation game}
 between two players, Attacker and Defender.
Two states in a labelled transition system are bisimilar if and only if Defender has a winning strategy, see e.g.\ \cite{SrbaVisiblyPDA:2009}.

We define some notation.
For relations $R \subseteq U \times 2^V$ and $S \subseteq V \times 2^W$, we define $(R \bullet S) \subseteq U \times 2^W$ by $R \circ \lift{S}$,
 where $\lift{S} := \{ ( \{ v_1, \ldots, v_k\}, \bigcup_{i=1}^k A_i ) \mid k \ge 0 \ \land \ (v_i,A_i) \in S\} \subseteq 2^V \times 2^W$
 and $\mathord{\circ}$ stands for the join of two relations.
Note that $\emptyset \mathrel{\lift{S}} \emptyset$, hence $u \mathrel{R} \emptyset$ implies $u \mathrel{(R \bullet S)} \emptyset$.
To avoid notational clutter in the following,
 if $C$ and $D$ are sets with $c \in C$ and $d \in D$, we often write $C D$ instead of $C \times D$ and $c d$ instead of $(c,d)$.

For finite sets of configurations $C, C' \subseteq Q \Gamma^*$ we call a relation
 $F \subseteq C C \times 2^{C' C'}$ a \emph{$(C,C')$-forcing relation}
 if  $c d \mathrel{F} S$ implies that Attacker, starting in~$c d$,
  can play so that he either wins or reaches a configuration in~$S$ (Defender may choose which configuration in~$S$).
If $F$ is a $(C,C')$-forcing relation, then %$F_{/\Gamma}$ defined by
 \[
  F_{/\Gamma} := \{ ( c X d Y , \{ c'_{1} X d'_{1} Y, \ldots, c'_{k} X d'_{k} Y \} ) \mid
   X,Y \in \Gamma \ \land \
     ( c   d , \{ c'_{1}   d'_{1}  , \ldots, c'_{k}   d'_{k} \} ) \in F  \}
 \]
 is a $(C \Gamma, C' \Gamma)$-forcing relation.
If $F$ is a $(C,C')$-forcing relation and $F'$ is a $(C',C'')$-forcing relation, then $F \bullet F'$ is a $(C,C'')$-forcing relation.
The union of $(C,C')$-forcing relations is a $(C,C')$-forcing relation,
 so there is a largest $(C,C')$-forcing relation.
We have that $p_0 X_0 \not\sim q_0 Y_0 \beta'$ holds if and only
 \[
  (p_0 X_0 q_0 Y_0) \mathrel{\hat F} \{p q \in Q Q \mid q \beta' \text{ has an outgoing transition} \}
 \]
 holds for the largest $(Q \Gamma, Q)$-forcing relation~$\hat F$.
(In particular, if $\beta' = \varepsilon$, then Attacker wins if and only if $(p_0 X_0 q_0 Y_0) \mathrel{\hat F} \emptyset$ holds.)
Hence it suffices to compute~$\hat F$ in exponential time.

For each $a \in \Sigmac$ we define a ``local'' $(Q \Gamma, Q \Gamma \Gamma)$-forcing relation~$\se{a}$ by
\begin{equation*} \label{eq:non-minimal}
\begin{aligned}
 (p X q Y) \mathrel\se{a} A \quad\Longleftrightarrow
  & \quad \exists\ p X \btran{a} p' X' X'' : A \supseteq \{ p' X' X'' q' Y' Y'' \mid q Y \btran{a} q' Y' Y'' \} \\
  & \lor \exists\  q Y \btran{a} q' Y' Y'' : A \supseteq \{ p' X' X'' q' Y' Y'' \mid p X \btran{a} p' X' X'' \}\,.
\end{aligned}
\end{equation*}
For $a \in \Sigmai$ and $a \in \Sigmar$ we analogously define local $(Q \Gamma, Q \Gamma)$-
 and $(Q \Gamma, Q)$-forcing relations $\se{a}$, respectively.
Those forcing relations can be computed in exponential time.
Let $\bar F$ be the least solution of the following equation system:
\begin{align*}
 F = \bigcup_{a \in \Sigmar} \se{a} \quad\cup\quad \bigcup_{a \in \Sigmai} \se{a} \bullet F \quad\cup\quad
   \bigcup_{a \in \Sigmac} \se{a} \bullet F_{/\Gamma} \bullet F\;.
\end{align*}
The least fixed point~$\bar F$ can be computed by a simple Kleene iteration starting from $F = \emptyset$.
The iteration terminates after at most $|Q \Gamma Q \Gamma \times 2^{Q Q}|$ rounds, each of which takes at most exponential time.
It is not hard to see that $\bar F$ is the largest $(Q \Gamma, Q)$-forcing relation.
% i.e., we have $p X \mathrel{\hat F} S$ if and only if there is $S' \subseteq S$ with $p X \mathrel{\bar F} S'$.
It follows that bisimilarity for (non-probabilistic) vPDA can be decided in exponential time.

%It will be crucial later on to make this computation slightly more efficient.
%To this end we replace the ``locally maximal'' relations~$\se{a}$ from above by smaller relations~$\sse{a}$ that
% on the right-hand sides do not contain some non-minimal sets which are uninteresting for Attacker.
%More precisely, analogously to~\eqref{eq:non-minimal},
% we define for each $a \in \Sigmac$ a $(Q \Gamma, Q \Gamma \Gamma)$-forcing relation~$\sse{a}'$ by
%\begin{equation*}% \label{eq:non-minimal}
%\begin{aligned}
% (p X q Y) \mathrel\sse{a} A \quad\Longleftrightarrow
%  & \quad \exists\ p X \btran{a} p' X' X'' : A = \{ p' X' X'' q' Y' Y'' \mid q Y \btran{a} q' Y' Y'' \} \\
%  & \lor \exists\  q Y \btran{a} q' Y' Y'' : A = \{ p' X' X'' q' Y' Y'' \mid p X \btran{a} p' X' X'' \}\,.
%\end{aligned}
%\end{equation*}
%For $a \in \Sigmai$ and $a \in \Sigmar$ we analogously define local $(Q \Gamma, Q \Gamma)$-
% and $(Q \Gamma, Q)$-forcing relations $\sse{a}$, respectively.

\newcommand{\hashr}{\#_r}
\newcommand{\hashi}{\#_{\mathit int}}
\newcommand{\hashc}{\#_c}

Now we consider a (probabilistic) pvPDA $\Delta = (Q,\Gamma,\Sigma,\mathord{\btran{}})$ with action partition $\Sigma = \Sigmar \cup \Sigmai \cup \Sigmac$.
We use essentially the reduction from Section~\ref{sec-prob-to-nondet}
 to compute a (non-probabilistic) vPDA $\Delta' = (Q,\Gamma',\Sigma',\mathord{\ctran{}})$,
 but we need to adapt it slightly to preserve ``visibly-ness'':
Instead of the action~$\#$ we need three actions $\hashr \in \Sigmar'$ and $\hashi \in \Sigmai'$ and $\hashc \in \Sigmac'$ in~$\Delta'$.
This change does not affect the correctness of the reduction.
Observe that $\Sigmar' = \{\hashr\}$ and $\Sigmac' = \{\hashc\}$ and $\Sigmai' = \Sigma \cup W \cup \{\hashi\}$.
For each $a \in \Sigmac$ we define a local $(Q \Gamma, Q \Gamma \Gamma)$-forcing relation~$\se{a}$ in~$\Delta'$ by
 \[
  \se{a} := \se{a}' \bullet \left( \bigcup_{w \in W} \se{w}' \right) \bullet \se{\hashc}'\,,
 \]
 where $\se{\cdot}'$ refers to the local forcing relation~$\se{\cdot}$ defined above,
  where the definition is applied to the (non-probabilistic) vPDA~$\Delta'$.
For $a \in \Sigmai$ and $a \in \Sigmar$ we analogously define local $(Q \Gamma, Q \Gamma)$-
 and $(Q \Gamma, Q)$-forcing relations $\se{a}$, respectively.
The fact that these are valid forcing relations in~$\Delta'$ follows from the structure of the reduction,
 where each transition is mapped to three consecutive transitions in~$\Delta'$.

\begin{lemma}
 For all $a \in \Sigma$, the relation $\se{a}$ can be computed in exponential time.
\end{lemma}
\begin{proof}[Proof of the lemma]
We assume $a \in \Sigmac$; the other cases are similar.
It suffices to show that, given $pX qY \in Q\Gamma Q\Gamma$ and $A \subseteq {Q\Gamma\Gamma Q\Gamma\Gamma}$,
 we can check in exponential time whether $pX qY \mathrel\se{a} A$ holds.
To show this we give an alternating PSPACE algorithm that checks whether $pX qY \mathrel\se{a} A$ holds.
Then the lemma follows from APSPACE = EXPTIME.
We formulate the APSPACE algorithm in terms of an existential player (corresponding to Attacker) and a universal player (corresponding to Defender):

%\medskip
\newcommand{\ind}{\mbox{}\quad}
\newcommand{\indd}{\ind\quad}
\newcommand{\inddd}{\indd\quad}
\newcommand{\indddd}{\inddd\quad}
\begin{minipage}{\textwidth}
 input: $a \in \Sigmac$ and $pX qY \in Q\Gamma Q\Gamma$ and $A \subseteq {Q\Gamma\Gamma Q\Gamma\Gamma}$ \\
 return: whether $pX qY \mathrel\se{a} A$ holds \\
 \ind ex.~player: choose either:\\
 \indd ex.~player: choose $d$ s.t.\ $pX \ctran{a} p \stacksymba{d}$ \\
 \indd un.~player: choose $e$ s.t.\ $qY \ctran{a} q \stacksymba{e}$ \\
 \ind or: \\
 \indd ex.~player: choose $e$ s.t.\ $qY \ctran{a} q \stacksymba{e}$ \\
 \indd un.~player: choose $d$ s.t.\ $pX \ctran{a} p \stacksymba{d}$ \\
 \ind ex.~player: choose either:\\
 \indd ex.~player: choose $w, T$ s.t.\ $p \stacksymba{d} \ctran{w} p \stacksymba{T}$ \\
 \indd un.~player: choose $U$ s.t.\ $q \stacksymba{e} \ctran{w} q \stacksymba{U}$ \\
 \ind or: \\
 \indd ex.~player: choose $w, U$ s.t.\ $q \stacksymba{e} \ctran{w} q \stacksymba{U}$ \\
 \indd un.~player: choose $T$ s.t.\ $p \stacksymba{d} \ctran{w} p \stacksymba{T}$ \\
 \ind ex.~player: choose either:\\
 \indd ex.~player: choose $p'X'X''$ s.t.\ $p \stacksymba{T} \ctran{\hashc} p'X'X''$ \\
 \indd un.~player: choose $q'Y'Y''$ s.t.\ $q \stacksymba{U} \ctran{\hashc} q'Y'Y''$ \\
 \ind or: \\
 \indd ex.~player: choose $q'Y'Y''$ s.t.\ $q \stacksymba{U} \ctran{\hashc} q'Y'Y''$ \\
 \indd un.~player: choose $p'X'X''$ s.t.\ $p \stacksymba{T} \ctran{\hashc} p'X'X''$ \\
 \ind return whether $p'X'X''q'Y'Y'' \in A$ holds
\end{minipage}
\end{proof}

%Since $\Gamma' \supseteq \Gamma$ and $\Sigma' \supseteq \Sigma$ have exponential size in~$\Delta$,
% computing those relations~$\se{a}$ takes exponential time.
We can compute the largest $(Q \Gamma, Q)$-forcing relation~$\hat F$ as above,
 i.e., by solving the equation system
\begin{align*}
 F = \bigcup_{a \in \Sigmar} \se{a} \quad\cup\quad \bigcup_{a \in \Sigmai} \se{a} \bullet F \quad\cup\quad
   \bigcup_{a \in \Sigmac} \se{a} \bullet F_{/\Gamma} \bullet F\
\end{align*}
using simple Kleene iteration.
As above, the iteration terminates after at most $|Q \Gamma Q \Gamma \times 2^{Q Q}|$ rounds, each of which takes at most exponential time.
This completes the proof.
\end{proof}

}{}

\end{document}